\documentclass[12pt]{iopart}

\usepackage{iopams}  
\usepackage{amsmath}
\usepackage{bbm}

\usepackage{geometry}
\usepackage[utf8]{inputenc}

\usepackage{amsmath}
\usepackage{amsfonts,amssymb,amsthm}    
\usepackage{mathtools}
\usepackage[english]{babel}                      

\usepackage{ifthen}    
\usepackage{multirow}
\usepackage{lscape}
\usepackage{bm}
\usepackage[mathcal]{euler}

\newtheoremstyle{colon}%
{}{}{\itshape}{}{\bfseries}{.}{ }{\thmname{#1}\thmnumber{ #2}\thmnote{\textbf{ (#3)}}}

\theoremstyle{colon}
\newtheorem{theorem}{Theorem}
\newtheorem{lemma}[theorem]{Lemma}
\newtheorem{coro}[theorem]{Corollary}
\newtheorem{definition}[theorem]{Definition}
\newtheorem{proposition}[theorem]{Proposition}

\usepackage[bbgreekl]{mathbbol}
\usepackage{amsfonts,amssymb,amsthm,mathtools}    
\DeclareSymbolFontAlphabet{\mathbbl}{bbold}
\DeclareSymbolFontAlphabet{\mathbbm}{bbold}
\DeclareSymbolFontAlphabet{\mathbb}{AMSb}%
\usepackage{graphicx}
\usepackage{dcolumn}
\usepackage{bm}
\usepackage{mathrsfs}
\usepackage[usenames]{color}
\usepackage[dvipsnames, table]{xcolor}    
\usepackage{accents} 
\usepackage{tensor}
\usepackage{tikz}
\usepackage[breaklinks=true,backref=none]{hyperref} 
\usepackage{todonotes}

\newcommand{\corurl}{BrickRed}  \newcommand{\corcite}{red}
\newcommand{\corlink}{blue}    \newcommand{\corfile}{black}

\hypersetup{
	linktocpage,
	colorlinks,
	urlcolor=\corurl,	
	citecolor=\corcite,
	linkcolor=\corlink,
	filecolor=\corfile,
	pdfnewwindow,
}
\usetikzlibrary{matrix}

\def\ledgee{{\setbox0\hbox{\ensuremath{\mathrel{\cdot}}}\rlap{\hbox to \wd0{\hss\ensuremath\wedge\hss}}\box0}}

\usepackage{graphicx}
\makeatletter
\newcommand\rrule[3][0pt]{%
	\ifdim#2>#3\math@hrule[#1]{#2}{#3}\else\math@vrule[#1]{#2}{#3}\fi}
\newcommand\math@hrule[3][0pt]{%
	\gdef\mystery@factor{0.07}%
	\@tempdima=#3%
	\rule[#1]{0pt}{#3}
	\raisebox{.5\@tempdima+#1}{%
		\makebox[#2][l]{\kern-.5\@tempdima\@@mathrule{#2}{#3}}}%
}
\newcommand\math@vrule[3][0pt]{%
	\gdef\mystery@factor{0.0}%
	\@tempdima=#2%
	\rule[#1]{0pt}{#3}
	\raisebox{-.0\@tempdima+#1}{%
		\kern0.5\@tempdima%
		\rotatebox{90}{\kern-0.5\@tempdima\makebox[#3][l]{\@@mathrule{#3}{#2}}}%
		\kern0.5\@tempdima}%
}
\def\@@mathrule#1#2{%
	\@tempdimb=#2%
	\@tempdima=\dimexpr#1-\mystery@factor\@tempdimb
	\pdfliteral{%
		q []0 d %
		1 J 
		\strip@pt\@tempdimb\space w \strip@pt\@tempdimb\space 0 m %
		\strip@pt\@tempdima\space 0 l S Q }}
\makeatother

\graphicspath{{./font/}{./}}

\def\wwedgee{{\setbox0\hbox{\ensuremath{\mathrel{\wedge}}}\rlap{\hbox to \wd0{\hss\,\ensuremath\wedge\hss}}\box0}}

\def\ledgee{{\setbox0\hbox{\ensuremath{\mathrel{\cdot}}}\rlap{\hbox to \wd0{\hss\ensuremath\wedge\hss}}\box0}}

\newcommand{\dd}{\mathchoice
	{\mathbbm{d}\rrule{.087ex}{1.605ex}\hspace*{0.15ex}} 
	{\mathbbm{d}\rrule{.087ex}{1.605ex}\hspace*{0.15ex}} 
	{\mathbbm{d}\rrule{.08ex}{1.125ex}\hspace*{0.15ex}}  
	{\mathbbm{d}\rrule{.06ex}{.8ex}\hspace*{0.15ex}}     
}

\renewcommand{\d}{\mathrm{d}}
\newcommand{\st}{\ \Big\rvert \ }

\newlength{\alturad}

\newcommand{\ww}{\wedge\hspace*{-5mm}\wedge\,\,}

\newcommand{\sdots}{\!\cdot\!\cdot\!\cdot\!}

\usepackage{titlesec}
\usepackage{slashed}

\titlespacing\section{0pt}{10pt plus 4pt minus 2pt}{1pt plus 2pt minus 2pt}
\titlespacing\subsection{0pt}{12pt plus 4pt minus 2pt}{1pt plus 2pt minus 2pt}
\titlespacing\subsubsection{0pt}{12pt plus 4pt minus 2pt}{0pt plus 2pt minus 2pt}

\newcommand{\Ev}{\mathrm{E}\hspace*{-.3pt}\mathrm{v}}

\begin{document}

\title{Poisson brackets in Sobolev spaces: a mock holonomy-flux algebra}

\author{J Fernando Barbero G$^{1, 2}$, Marc Basquens$^3$, Bogar Díaz$^{2, 3}$, and Eduardo J S Villase\~nor$^{2, 3}$}
\address{$^1$ Instituto de Estructura de la Materia, CSIC, Serrano 123, 28006 Madrid, Spain}
\address{$^2$ Grupo de Teor\'{\i}as de Campos y F\'{\i}sica Estad\'{\i}stica, Instituto Gregorio Mill\'an (UC3M), Unidad Asociada al Instituto de Estructura de la Materia, CSIC}
\address{$^3$ Departamento de Matem\'aticas, Universidad Carlos III de Madrid, Avda.\  de la Universidad 30, 28911 Legan\'es, Spain}
\ead{fbarbero@iem.cfmac.csic.es, marc.basquens.munoz@gmail.com, bodiazj@math.uc3m.es, ejsanche@math.uc3m.es}


\begin{abstract}
The purpose of this paper is to discuss a number of issues that crop up in the computation of Poisson brackets in field theories. This is specially important for the canonical approaches to quantization and, in particular, for loop quantum gravity. We illustrate the main points by working out several examples. Due attention is paid to relevant analytic issues that are unavoidable in order to properly understand how computations should be carried out. Although the functional spaces that we use throughout the paper will likely have to be modified in order to deal with specific physical theories such as general relativity, many of the points that we will raise will also be relevant in that context. The specific example of the \textit{mock holonomy-flux algebra} will be considered in some detail and used to draw some conclusions regarding the loop quantum gravity formalism.
\end{abstract}

%
%
%
%
%
%
%
\section{Introduction}{\label{sec_defs}}

Loop quantum gravity (LQG) is one of the leading approaches to the quantization of general relativity (GR). In the last three decades it has blossomed into a multi-pronged line of research encompassing Hamiltonian (``$3+1$'' or ``canonical'') methods and covariant ideas (spin foams). The starting point of LQG was the Hamiltonian description of GR in the phase space of an $SU(2)$ Yang-Mills theory found by Ashtekar \cite{Ashtekar1,Ashtekar2}. In marked contrast with the ADM \mbox{geometrodynamical} formalism---where the basic field is the metric---this formulation relies on \textit{connections} as the basic configuration gravitational variables. This makes it possible to import many ideas and techniques, developed for the gauge theories used in particle physics, to deal with the difficult problem of quantizing gravity.

One of the pillars of canonical LQG is the so-called \textit{holonomy-flux algebra} introduced by Ashtekar, Corichi, and  Zapata in \cite{ACZ}, which replaced the original $T$-algebra proposed by Rovelli and Smolin in \cite{HFA}. The holonomy-flux algebra is defined by two types of basic phase space functions, respectively associated with curves and surfaces in a 3-dimensional differential manifold $\Sigma$. The Poisson brackets of these objects play an essential role in LQG because the first step in the quantization process is to represent the algebra defined by them in a suitable Hilbert space. As often argued in the literature (see, for instance, \cite{ACZ}) this is easier said than done because the actual computation of the Poisson brackets is somewhat problematic. For instance, the \textit{distributional character} of the holonomies and fluxes seems to lead to inconsistencies, such as violations of the Jacobi identity, if the fluxes are assumed to Poisson-commute as they apparently should. The root of the problem is that the fluxes and holonomies are not differentiable functions in phase space and, hence, a direct computation of their Poisson brackets makes no sense. This problem can be circumvented in two ways: i) by defining the Poisson brackets of the fluxes and holonomies with the help of differentiable functions (a regularization process) \cite{Thiemann1, CattaPer} or  ii) by associating the fluxes to (singular) vector fields \cite{ACZ}. The algebra obtained by these methods has become the foundation of LQG in its canonical incarnation (see \cite{Thiemann1} for a careful discussion on how the modified algebra of classical basic functions can be obtained by a precise regularization process).

According to their definition, the fluxes depend only on one of the canonical variables (the momentum associated with the connection) and, hence, their Poisson brackets should be zero. However, the modified algebra mentioned above changes this: the fluxes do not Poisson-commute. This seemingly contradictory fact has been justified in \cite{CattaPer} by invoking the necessity of taking into account the Gauss law (the constraint that generates internal $SU(2)$ transformations) and the possibility of working with differentiable basic variables which coincide with the fluxes on the phase space submanifold defined by the constraints, but depend both on connections and momenta.
In fact, as discussed in \cite{Thiemann1, Freidel_2013}, by relying on covariance arguments it is possible to argue that the correct definition of the fluxes must involve \textit{both} configuration variables and momenta. Of course, a fruitful and useful attitude is just to postulate the desired algebra, check its consistency and carry on.


The central goal of the present paper is to highlight the fact that the actual computation of Poisson brackets cannot be properly done without paying attention to some mathematical issues. For instance, it only makes sense to compute the Poisson brackets of \textit{differentiable functions} (as remarked in \cite{CattaPer}). Of course, to properly talk about differentiability it is necessary to work in the appropriate mathematical setting and, as we will argue later, with a clear definition of this concept. Furthermore, in order for the Jacobi identity to make sense, the Poisson brackets of the relevant phase space functions appearing in it (some of them Poisson brackets themselves) \textit{must also be differentiable}.  The appearance of distributions in the basic Poisson brackets of canonical variables in field  theories must be properly understood (beyond the simple use of smearings). Finally, the properties of the symplectic form$\Omega$, in particular the fact that it is \textit{non-degenerate} and \textit{closed}, guarantee that the Jacobi identity will hold if it can be computed. In other words, if the Poisson brackets are defined and are smooth enough, \textit{they will always satisfy the Jacobi identity}.

In order to discuss the previous points, we will work out in detail a number of examples. In all the cases, we will introduce the necessary mathematical structures, but no more. Our purpose is to provide some sample computations that highlight the main points without unnecessary complications. Specifically, we will take care of functional analytic and topological issues by working always in Sobolev spaces (in fact, $W^{s,2}$ Hilbert spaces involving $s$ derivatives). We will consider functions defined on spatial manifolds $\Sigma$ of different dimensions with or without boundary. When considering boundaries we will carefully show how (and if) differentiability is affected by them. Functional derivatives will be carefully introduced, in particular, the expressions for Poisson brackets in terms of these objects will be discussed in detail. We will define several phase space functions and show how to compute their Poisson brackets. Among these functions, evaluations will play an important role. Finally, we will introduce \textit{mock holonomy-flux variables} and study their algebra.

We want to emphasize, from the beginning, that we are not claiming that the functional framework that we employ in this paper can be used or adapted to solve the problem of quantizing gravity (we do not know how to do it at this point and we cannot exclude the possibility that it is impossible). Our intention is just to highlight a number of basic mathematical points in order to shed light on some issues that have been a matter of contention and, in our opinion, a source of confusion.

The structure of the paper is the following. After this introduction, we  provide the minimal mathematical background appropriate for our purposes in Section \ref{sec_mat_background}. Several sample functional spaces  that we will use throughout the paper will be described in Section \ref{sec_sample_spaces}. We will illustrate some interesting points regarding differentiability and the computation of Poisson brackets, both in manifolds with and without boundary, in Section \ref{sec_one_dim_examples}. The mock holonomy flux algebra will be discussed in Section \ref{sec_mock_HFalgebra}. We end the paper with a short discussion of the consequences of our results and some comments in \ref{sec_conclusions}. Some derivations have been relegated to the Appendices. The notation used throughout the paper is standard.

%
%
\section{Mathematical background}{\label{sec_mat_background}}

The configuration space for a field theory consists of functions defined on a differential manifold $\Sigma$ that plays the role of space. These functions are usually subject to regularity conditions, required in part by the fact that they have to satisfy dynamical equations involving differential operators. In addition, it may be necessary to introduce some topological structure in the configuration space in order to have proper notions of continuity and smoothness. Although, in physics, a traditional and often justifiable attitude is that of assuming that the fields are \textit{as nice as needed} to guarantee that the mathematical expressions where they appear make sense, it is often necessary to phrase these requirements explicitly.

\subsection{Functional spaces}{\label{subsec_funct}}

The configuration spaces that we will use in the paper will be Sobolev spaces. Strictly speaking their elements are not functions but, rather, equivalence classes of functions, so we will work with a slight generalization of the physical concept of field. Although this will force us to be careful at times, no major difficulties will be encountered. Essentially all the following definitions, theorems, and their proofs can be found in \cite{Brezis}.

\medskip

\begin{definition}
Let $I$ be an open interval of the real line $\mathbb{R}$, bounded or unbounded, and let $C_c^1(I)$ denote the space of continuously differentiable real functions with compact support on $I$ (test functions). We define
\begin{equation}\label{Sobolev_1dim}
H^1(I):=\Big\{ u\in L^2(I) \st \exists g\in L^2(I) \ : \ \    \int_I u \varphi'=-\int_I g\varphi \ , \ \forall \varphi\in C_c^1(I) \Big\}\ .
\end{equation}
\end{definition}

In words, the elements of $H^1(I)$ are $L^2(I)$ ``functions'' with square integrable \textit{weak derivatives} $g$ defined on test functions by integration by parts. We will usually write $g=u'$.
\medskip

For our purposes the following properties and results about $H^1(I)$ will be useful:

\medskip

\begin{theorem}
\label{th_H1_Hilbert}
With the scalar product
\begin{equation}\label{escalar_prod_1dim}
\langle u,v\rangle_{H^1}:=\int_I\big(uv+u'v'\big)\,,
\end{equation}
and its associated norm $\|u\|^2_{H^1}:=\langle u,u\rangle_{H^1}$, the space $H^1(I)$ is a separable Hilbert space.
\end{theorem}

\medskip

\begin{theorem}
\label{th_caclulus}
Let $u\in H^1(I)$ and $I$ a bounded or unbounded interval of $\mathbb{R}$; then there exists a \emph{unique, continuous} function $\tilde{u}\in C(\bar{I})$ such that $u=\tilde{u}$ almost everywhere, and
\begin{align*}
\tilde{u}(b)-\tilde{u}(a)=\int_b^a u'(t)\mathrm{d}t\,,\quad \forall b,a\in\bar{I}\,.
\end{align*}
\end{theorem}
The proof can be found in reference \cite{Brezis}, Theorem 8.2 and Remark 5. As the elements of $H^1(I)$ have continuous representatives in the closure $\overline{I}$, it makes sense to talk about the values of $u\in H^1(I)$ \textit{at any point} $x\in \overline{I}$ despite the fact that, strictly speaking, the elements of $H^1(I)$ are defined only ``modulo zero measure sets''. Notice that, in particular, the boundary values of $u$ are well defined if the interval $I$ is bounded. It is also important to point out that the continuous representative $\tilde{u}$ of $u\in H^1(I)$ is, actually, differentiable a.e. and the classical derivative is equal to the weak derivative a.e. \cite{Evans}.

\medskip

\begin{theorem}
\label{th_parts}
Let $u,v\in H^1(I)$. Then
\begin{align*}
uv\in H^1(I)
\end{align*}
\textit{and}
\begin{align*}
(uv)'=u'v+uv'\,.
\end{align*}
\textit{Furthermore, the formula for integration by parts holds:}
\begin{align*}
\int_a^bu'v=\tilde{u}(b)\tilde{v}(b)-\tilde{u}(a)\tilde{v}(a)-\int_a^b u v'\,,\quad\forall a,b\in\overline{I}\,.
\end{align*}   
\end{theorem}

As a consequence, $H^1(I)$ is a Banach algebra.

\bigskip

The Sobolev space $H^1(I)$ introduced above can be generalized in multiple ways, for instance, by introducing higher order derivatives or replacing the interval $I$ by regions in $\mathbb{R}^n$, with $n\in\mathbb{N}$. As done above, we will forgo generality and just introduce the spaces that we will use in the rest of the paper. 

We start by defining the Sobolev space
\[
H^1(\mathbb{R}^3)=\bigg\{ u\in L^2(\mathbb{R}^3) \st \forall i \,,\, \exists g_i\in L^2(\mathbb{R}^3)\   :\int_{\mathbb{R}^3}\!\!\!u\frac{\partial\varphi}{\partial x_i}\!=\!-\int_{\mathbb{R}^3}\!\!g_i\varphi  , \ \forall\varphi\in C_c^\infty(\mathbb{R}^3)  \bigg\}\,.
\]
For $u\in H^1(\mathbb{R}^3)$ we write $g_i=\frac{\partial u}{\partial x_i}$  and in the following we use the shorthand $\partial_i\!=\!\frac{\partial\,\,\,\,\,}{\partial x_i}$, $i=1,2,3$. We can endow $H^1(\mathbb{R}^3)$ with the structure of a Hilbert space by introducing the scalar product
\begin{equation}\label{scalar3_1}
  \langle u,v\rangle:=\int_{\mathbb{R}^3}\big(uv+\sum_{i=1}^3\left(\partial_i u\right)\,\left(\partial_i v\right)\big)\ .
\end{equation}

When considering field theories in 3-dimensional spatial regions and, in particular, in our discussion of the holonomy-flux algebra, we will use the Hilbert space $\mathbbm{H}^2(\mathbb{R}^3)$ obtained by endowing the set
\begin{equation}\label{Htilde2}
  H^2(\mathbb{R}^3):=\left\{u\in H^1(\mathbb{R}^3) \st \frac{\partial u}{\partial x_i}\in H^1(\mathbb{R}^3) \ , \ \forall i=1,2,3  \right\}
\end{equation}
with the scalar product
\begin{equation}\label{scalar3_2}
  \langle u,v\rangle_{\mathbbm{H}^2}:=\int_{\mathbb{R}^3}\Big(uv+2\sum_{i=1}^3\left(\partial_i u\right)\,\left(\partial_iv\right)+\sum_{i, j=1}^3\left(\partial_i\partial_j u\right)\, \left( \partial_i\partial_j v \right)\Big)\,.
\end{equation}
This scalar product differs from the usual one owing to the presence of the 2 factor in the second term. However, it is possible to show that the associated norm $\|\cdot\|_{\mathbbm{H}^2}$ is equivalent to the standard one and, hence, nothing changes from the topological point of view. Indeed, the standard norm is
\begin{align*}
\|f\|_s^2=\int_{\mathbb{R}^3} \Big(f^2+\sum_{i=1}^3\left(\partial_i f\right)^2 +\sum_{i, j=1}^3\left(\partial_i\partial_j f\right)^2 \Big)\ ,
\end{align*}
and it is immediate to see that $\|f\|_s\leq \|f\|_{\mathbbm{H}^2}\leq 2^{1/2}\|f\|_s$.

An important consequence of the fact that we will always work with Hilbert spaces is the very useful characterization of the elements of their duals furnished by the Riesz-Fréchet representation theorem that we state in the following form

\medskip

\begin{theorem}
\label{th_Riesz_Frechet}
{Let $F:\mathcal{H}\rightarrow \mathbb{R}$ be a linear map on a real Hilbert space $\mathcal{H}$, then $F$ is continuous if and only if there exists $\psi\in \mathcal{H}$ such that
\begin{align*}
F(v)=\langle\psi,v\rangle_{\mathcal{H}}\,, \forall v\in \mathcal{H}\,.
\end{align*}
Furthermore, $\psi$ is unique and $\|F\|_{\mathcal{H}^*}=\|\psi\|_{\mathcal{H}}$\,.
}
\end{theorem}

\medskip

\noindent In the following, we will mostly deal with real, separable Hilbert spaces. In such circumstances it will be always possible to find a countable orthonormal basis $(e_n)_{n\in \mathbb{N}}$ and expand any vector $v\in\mathcal{H}$ as
\begin{align*}
v=\sum_{n=1}^{\infty}\langle v,e_n\rangle_{\mathcal{H}}\,e_n\,,
\end{align*}
hence,
\begin{align*}
F(v)=\sum_{n=1}^{\infty}\langle v,e_n\rangle_{\mathcal{H}}F(e_n)=\langle \psi,v\rangle_{\mathcal{H}}
\end{align*}
with
\begin{equation}\label{Riesz}
\psi=\sum_{n=1}^{\infty}F(e_n)e_n\,.
\end{equation}
This expansion is often useful to find the Riesz-Fréchet representative of a continuous, linear functional (see \ref{app_1}). It is also worth pointing out that a necessary condition for a linear functional $F:\mathcal{H}\rightarrow \mathbb{R}$ to be continuous is
\begin{equation}\label{nec_cond_Riesz}
\sum_{n=0}^{\infty}|F(e_n)|^2<+\infty
\end{equation}
for any orthonormal basis $(e_n)_{n\in \mathbb{N}}$.



\subsection{Differentiability and functional derivatives}\label{subsec_diff}

Somewhat surprisingly, the word \textit{differentiability} in the mathematical context and in the Hamiltonian treatment of gravity (and field theories, in general) has different meanings, in particular when spacetime boundaries are present.  

From a mathematical perspective, a very useful notion is what we will here refer to as \textit{Fréchet differentiability} (or differentiability for short). This is a fundamental concept in analysis upon which many important theorems and results are based. 

In the context of field theories, especially those formulated in manifolds with boundary, the term \textit{differentiability} is somewhat more vague, but it is usually employed to indicate that the variations of certain functionals (given by integrals) do not have boundary contributions. The justification for this is the necessity (discussed by Regge and Teitelboim in \cite{RT}) to take into account boundary terms to guarantee that the solutions to the Hamiltonian version of the Einstein field equations have the correct behaviour in the asymptotically flat case. In fact, the Regge-Teitelboim procedure gives the correct Hamiltonian description in this setting. 

In our opinion this dichotomy has caused some confusion that we need to dispel in order to correctly frame some of the ideas that we present in the paper. 

We start by recalling the main results about differentiability in the mathematical sense.

\medskip

\begin{definition}[Fréchet differentiability]
\label{def_differentiability}
Let $\mathcal{B}_1$ and $\mathcal{B}_2$ be two Banach spaces with norms $\|\cdot\|_1$ and $\|\cdot\|_2$ respectively. Let $A\subset \mathcal{B}_1$ be open and let us consider a function
\begin{align*}
f:A\rightarrow \mathcal{B}_2\,.
\end{align*}
We say that the function $f$ is \emph{Fr\'echet differentiable} at $x\in A$ if there exists a linear and continuous map $\dd_xf:\mathcal{B}_1\rightarrow \mathcal{B}_2:h\mapsto\dd_xf(h)$, called the differential of $f$ at $x$, such that
\begin{align*}
\lim_{h\rightarrow 0}\frac{\|f(x+h)-f(x)-\dd_xf(h)\|_2}{\|h\|_1}=0\,.
\end{align*}
\end{definition}
It can be immediately shown that, when it exists, the differential is unique. A particular simple instance of differentiable functions are those linear and continuous as can be seen by replacing $\dd_xf(h)$ by $f(h)$ in the definition.

\medskip

Differentiability is a fruitful concept that allows us to extend many results of the analysis in $\mathbb{R}^n$, $n\in\mathbb{N}$ (for instance, the chain rule) to infinite dimensional Banach spaces such as the ones that we use in the paper. Also, this is the kind of differential used to define the exterior derivative in differential geometry.

\bigskip

In the context of gravitational theories and, more generally, in the discussion of the Lagrangian and Hamiltonian formulations of field theories defined on manifolds with boundaries, the word \textit{differentiability} is sometimes used to refer to the possibility of writing the variation of a functional $S[\phi]$ depending on fields $\phi$ (usually tensors on a manifold $M$) in the form
\begin{equation}\label{Wald_funct}
\delta S=\int_M \frac{\delta S}{\delta \phi} \delta \phi\,,
\end{equation}
requiring, in particular, that \textit{no boundary integrals appear in the right hand side}. Notice that, this ``definition'' is rather vague from the mathematical point of view as, for instance, no mention is made about the functional space where the field $\phi$ lives nor the regularity properties satisfied by $S$.

We will sometimes refer to this other concept as Regge-Teitelboim or RT-differentiability (also RT-admissibility) \cite{soloviev,RT,Faddeev}. The formal object $\frac{\delta S}{\delta \phi}$ that appears in \eqref{Wald_funct} is usually referred to as the \textit{functional derivative of }$S$. Several comments are in order:
\begin{itemize}
\item The RT-admissibility condition guarantees that the variational equations coming from the action do not have boundary contributions that might clash with the boundary conditions imposed on the fields. This is very important in the case of asymptotically flat general relativity as discussed in \cite{RT}.
\item The Fr\'echet and RT differentiability concepts have very little in common. The first one has to do with the problem of finding suitable approximations to functions at points of their domains whereas the other is tied to the dynamics defined by an action principle.  Let us, though, try to compare them. The field variations in \eqref{Wald_funct} play the role of the $h$ in Definition \ref{def_differentiability}, hence, in the case of a scalar function $S$, the variation $\delta S$ can be interpreted as the ``action'' of a ``dual object'' of the form
    \begin{align*}
    \int_M \frac{\delta S}{\delta \phi}(\cdot)
    \end{align*}
    on $h$. This action looks like the Fr\'echet differential in $L^2$ when one uses the Riesz-Fr\'echet theorem, so it may be natural to consider the fields as elements of $L^2$. However, this is problematic because, as general elements of $L^2$, it is then impossible to talk neither about field derivatives nor field values at boundaries (which we need for some of our arguments).
\item The mathematical consequences of Fréchet differentiability are clear and many theorems rely on this concept. It is quite dangerous to export these results to situations in which differentiability is understood in the second sense.
\item A situation where the two preceding comments apply is the one contemplated in this paper: the computation of Poisson brackets in field theories. In the following, the relevant concept will be Fréchet-differentiability used in conjunction with the mathematical properties of the Sobolev spaces that we use as configuration manifolds (and their cotangent bundles which will be our phase spaces).
\end{itemize}
In order to avoid any confusion we spell now the definition of functional derivative that we will use throughout the paper

\medskip

\begin{definition}[Functional derivative]
\label{def_functional_derivative}
Let us consider a differentiable function $F:H\rightarrow \mathbb{R}$ on a Sobolev Hilbert space $H$. The functional derivative of $F$ at $\psi\in H$, denoted as $DF(\psi)$, is the unique element of $H$ satisfying
\begin{align*}
\dd_\psi F(h)=\langle DF(\psi),h\rangle_{H}\,,
\end{align*}
for all $h\in H$.
\end{definition}

\medskip

\noindent Notice that this is well defined as a consequence of the Riesz-Fr\'echet representation Theorem \ref{th_Riesz_Frechet} and the fact that the differential is a continuous linear functional. We will say that the functional derivative $DF(\psi)$ is the \textit{Riesz-Fr\'echet representative} of $\dd_\psi F$. Whenever convenient we will drop the point $\psi$ where the differential is computed. If $H$ is an $L^2$ space, the functional derivative can be understood in the sense of \eqref{Wald_funct} but we consider other situations.

\subsection{The phase space and Poisson brackets}\label{subsec_phase space}

The configuration spaces for field theories are function spaces endowed with the appropriate topological structures. Although the tangent bundle of such a configuration space would be the natural setting to describe the Lagrangian dynamics of a field system, quite often it is better to employ the so called \textit{manifold-domains}. These can be loosely described as bundles for which the basis manifold is a given function space, but the fibers consist of elements of a \textit{different} (larger) function space \cite{ChernoffMarsden}. For instance, due to the presence of spatial derivatives, the appropriate domain for the standard Lagrangian of a scalar field is $T_{H^1}L^2=H^1\times L^2$, where the base manifold domain is $H^1$, but the fibers (field velocities) consist of elements of $L^2$ because no spatial derivatives of the velocities show up in the Lagrangian.  In the following, we will work with tangent or cotangent bundles defined on some particular function spaces. The readers should always keep in mind that the standard field theories \textit{are not} defined on these types of spaces.

\medskip

If the configuration space $Q$ of a field theory is a Sobolev space $H$ of one of the types that we are considering here, the cotangent bundle $T^*Q$ (phase space) is isomorphic to the Cartesian product $H\times H^*$. The elements of $T^*Q$ are pairs $(\phi,\bm{\pi})$ with $\phi\in H$ and $\bm{\pi}\in H^*$ (we will use boldface letters to denote covectors, i.e. elements of dual Hilbert spaces). As a consequence of the Riesz-Fréchet representation Theorem \ref{th_Riesz_Frechet} there is a unique $\pi\in H$ such that $\bm{\pi}(\cdot)=\langle \pi,\cdot\rangle_H$.

For $(\phi, \bm{\pi})\in T^*(Q)$ we will write tangent vectors sitting at this phase space point as
\begin{align*}
\big((\phi,\bm{\pi}),(X_1, {\bm{\mathrm{X}}}_2)\big)\in T_{(\phi,\bm{\pi})}T^*Q\,.
\end{align*}
It is important to notice that the second component of the vector is a covector. Vector fields are smooth assignments of vectors of the previous type to all the points of $T^*Q$. We will denote the space of such vector fields as $\mathfrak{X}(T^*Q)$.

The canonical symplectic form in $T^*Q$ acts on fields $\mathbb{X},\mathbb{Y}\in\mathfrak{X}(T^*Q)$ as
\begin{equation}\label{symplectic1}
\Omega(\mathbb{X},\mathbb{Y})={\bm{\mathrm{Y}}}_2(X_1)-{\bm{\mathrm{X}}}_2(Y_1)\,.
\end{equation}
This result can be obtained by a detailed computation starting from the symplectic potential $\theta$ canonically defined in the cotangent bundle $T^*Q$. A result by Marsden \cite{Marsden} tells us that in the cases considered here (all of them reflexive Banach manifolds) the symplectic form is strongly non-degenerate. By using the Riesz-Fréchet representatives of ${\bm{\mathrm{Y}}}_2$ and ${\bm{\mathrm{X}}}_2$ we find
\begin{equation}\label{symplectic2}
\Omega(\mathbb{X},\mathbb{Y})=\langle Y_2,X_1\rangle_{H}-\langle X_2,Y_1\rangle_{H}\,.
\end{equation}
An important consequence of the strong non-degeneracy of the symplectic form $\Omega$ \eqref{symplectic1} is the possibility to define the Hamiltonian vector field $\mathbb{X}_f$ associated with a real differentiable function $f$ in phase space as the unique solution to the equation
\begin{equation}\label{Poisson1}
\imath_{\mathbb{X}_f}\Omega=\dd f\,,
\end{equation}
where here $\dd$ is the exterior differential in $T^*Q$.

In order to write the symplectic form in terms of functional derivatives we need to introduce partial functional derivatives. For instance,  for real differentiable functions on $H\times H$ we first introduce the continuous maps
\begin{align*}
  & \sigma_1:H\rightarrow H\times H:h\mapsto (h,0)\,, \\
  & \sigma_2:H\rightarrow H\times H:h\mapsto (0,h)\,,
\end{align*}
and consider the continuous and linear functionals from $H$ to $\mathbb{R}$
\begin{align*}
  & (\dd_{(\psi_1,\psi_2)}g)\circ \sigma_1\,, \\
  & (\dd_{(\psi_1,\psi_2)}g)\circ \sigma_2\,,
\end{align*}
where $g:H\times H\rightarrow \mathbb{R}$ is a differentiable function (in terms of the natural norms in $H\times H$ and $\mathbb{R}$). We have now

\medskip

\begin{definition}
\label{def_partial_functional_derivative}
The partial functional derivatives of a differentiable real function $g:H\times H\rightarrow \mathbb{R}$ at $(\psi_1,\psi_2)\in H\times H$, denoted as $D_1g(\psi_1,\psi_2)$ and $D_2g(\psi_1,\psi_2)$, are, respectively, the Riesz-Fr\'echet representatives of the continuous and linear functionals
\begin{align*}
  & (\dd_{(\psi_1,\psi_2)}g)\circ \sigma_1\,, \\
  & (\dd_{(\psi_1,\psi_2)}g)\circ \sigma_2\,.
\end{align*}
\end{definition}

\noindent In other words
\begin{align*}
&\dd_{(\psi_1,\psi_2)}g(h,0)=\langle D_1g(\psi_1,\psi_2) , h\rangle_{H}\,,\\
&\dd_{(\psi_1,\psi_2)}g(0,h)=\langle D_2g(\psi_1,\psi_2) , h\rangle_{H}\,,
\end{align*}
for all $h\in H$. In general we have
\begin{equation}\label{dH}
\dd_{(\psi_1,\psi_2)}g(h_1,h_2)=\langle D_1g(\psi_1,\psi_2) , h_1\rangle_{H}+\langle D_2g(\psi_1,\psi_2) , h_2\rangle_{H}\,.
\end{equation}
Notice that $\displaystyle D_1g(\psi_1,\psi_2)\,, D_2g(\psi_1,\psi_2)\in H$. In the following, we will apply the preceding results to the computation of Poisson brackets. 

\medskip

\begin{definition}
Given two functions in phase space $f,g:T^*Q\rightarrow \mathbb{R}$ their Poisson bracket $\{f,g\}:T^*Q\rightarrow \mathbb{R}$ is
\begin{equation}\label{Poisson2}
\{f,g\}:=\Omega(\mathbb{X}_f,\mathbb{X}_g)=\dd f(\mathbb{X}_g)=-\dd g(\mathbb{X}_f)\,,
\end{equation}
where $\mathbb{X}_f$ and $\mathbb{X}_g$ are the Hamiltonian vector fields associated with $f$ and $g$. Here we suppose that $\Omega$ is strongly non-degenerate.
\end{definition}

\medskip

In order to compute the Hamiltonian vector field associated with a differentiable function $f:H \times H \rightarrow \mathbb{R}$ we have to find $\mathbb{X}_f$ such that
\begin{equation}\label{X_f}
\Omega(\mathbb{X}_f,\mathbb{Y})=\dd f(\mathbb{Y})\,,
\end{equation}
for all $\mathbb{Y}\in\mathfrak{X}(T^*Q)$. According to \eqref{symplectic2}, the left hand side of equation \eqref{X_f} is
\begin{align*}
\langle Y_2,X_{f\,1}\rangle_{H}-\langle X_{f\, 2},Y_1\rangle_{H}\,.
\end{align*}
whereas, according to \eqref{dH}, the right hand side is
\begin{align*}
\dd f(\mathbb{Y})=\langle D_1 f, Y_1\rangle_{H}+\langle D_2 f, Y_2\rangle_{H}\,.
\end{align*}
We immediately conclude that the components of the Hamiltonian vector field $\mathbb{X}_f$ are
\begin{align*}
X_{f\,1}=D_2 f\,,\quad X_{f\,2}=-D_1 f\,,
\end{align*}
and the Poisson bracket is defined by
\begin{equation}\label{PB_final}
\{f,g\}=\langle D_1 f, D_2 g\rangle_{H}-\langle D_1 g, D_2 f\rangle_{H}\,.
\end{equation}
The last equality gives the expression of the Poisson bracket of the phase space functions $f,g$ in terms of the functional derivatives that we have defined above.

\medskip

We end this section with the following comment. The Poisson brackets for a field theory (defined, say, in $\mathbb{R}^n$) are usually written as
\begin{align*}
\{f,g\}=\int_{\mathbb{R}^n}\left(\frac{\delta f}{\delta \phi(x)}\frac{\delta g}{\delta \pi(x)}-\frac{\delta g}{\delta \phi(x)}\frac{\delta f}{\delta \pi(x)}\right) \mathrm{d}^nx\,,
\end{align*}
which, formally, looks like a sum of $L^2$ scalar products of ``functional derivatives'', which are to be understood as ``evaluated at arbitrary but fixed $\phi$ and $\pi$'' so that the previous expression can be interpreted as defining a real function in phase space (i.e. a function of $\phi$ and $\pi$). This expression is usually thought of as a generalization of the standard formula for finite dimensional mechanical systems obtained by trading sums for integrals in such a way that, effectively, the integration variable $x$ becomes a ``continuous summation index''. As we have shown, there is an element of truth in this approach but it must be phrased in precise mathematical terms.

%
%
\section{Sample spaces}{\label{sec_sample_spaces}}

In this section, we introduce several sample functional spaces and study some interesting objects defined in them that will play an important role in the computations of Poisson brackets that we give in the next two sections.

\subsection{The \texorpdfstring{$H^1(0,1)$}{H1(0,1)} Sobolev space}

As a first example, we consider the Sobolev space $H^1(0,1)$ defined on the interval $(0,1)$ of the real line with the scalar product given by \eqref{escalar_prod_1dim}. We will use this space to illustrate some issues related to the computation of Poisson brackets in a field theory defined on a manifold with boundary. In the following, we will make use of the following interesting functions:

\medskip

\noindent \textbf{1) The evaluation:} Let us take $x\in[0,1]$ and consider the following function
\begin{align*}
\Ev_x:H^1(0,1)\rightarrow \mathbb{R}:u \mapsto \widetilde{u}(x)\,,
\end{align*}
where $\widetilde{u}$ denotes the unique, continuous representative of $u\in H^1(0,1)$ whose existence is guaranteed by Theorem \ref{th_caclulus}. We study now some properties of $\Ev_x$.
\begin{enumerate}
\item[i)] For every $u \in H^1(0,1)$ it is possible to show that $\Ev_x(u)$ can be written as a scalar product in $H^1(0,1)$, indeed, let $x\in [0,1]$ and define the function
\begin{equation}\label{Riesz_rep_eval}
\mathcal{E}_x:[0,1]\rightarrow\mathbb{R}:t\mapsto\mathcal{E}_x(t)=\left\{\begin{array}{ll}
                                                                 \displaystyle\frac{\cosh(1-x)\cosh t}{\sinh1}\,, & t\in[0,x] \\
                                                                 \quad&\quad \\ 
                                                                 \displaystyle\frac{\cosh x\cosh(1-t)}{\sinh1}\,, & t\in[x,1]
                                                               \end{array}\right.
\end{equation}
Then $\displaystyle \Ev_x(u)=\widetilde{u}(x)=\langle\mathcal{E}_x, u\rangle_{H^1}\,.$ Notice that the Riesz-Fr\'echet theorem guarantees the uniqueness of $\mathcal{E}_x$.

In order to prove this, one has first to check that $\mathcal{E}_x\in H^1(0,1)$. This is straightforward: on one hand we obviously have $\mathcal{E}_x\in L^2(0,1)$. On the other,
\begin{align*}
\mathcal{E}'_x(t)=\left\{\begin{array}{ll}
    \displaystyle\phantom{-}\frac{\cosh(1-x)\sinh t}{\sinh1}\,, & t\in[0,x) \\
    \quad&\quad\\
     \displaystyle-\frac{\cosh x\sinh(1-t)}{\sinh1}\,, & t\in(x,1]
                \end{array}\right.\,,
\end{align*}
which is also an element of $ L^2(0,1)$. Using now integration by parts (see Theorem \ref{th_parts}) we get
\begin{align*}
  \langle\mathcal{E}_x, u\rangle_{H^1}&=\int_0^1\big(u(t)\mathcal{E}_x(t)+u'(t)\mathcal{E}_x'(t)\big)\mathrm{d}t \\
  &= \int_0^x\big(u(t)\mathcal{E}_x(t)+u'(t)\mathcal{E}_x'(t)\big)\mathrm{d}t\\
  & \ \ \ \ +\int_x^1\big(u(t)\mathcal{E}_x(t)+u'(t)\mathcal{E}_x'(t)\big)\mathrm{d}t\\
  &=\widetilde{u}(t)\mathcal{E}'_x(t)\Big|_0^x+\widetilde{u}(t)\mathcal{E}'_x(t)\Big|_x^1+\int_0^1u(t)\big(\mathcal{E}_x(t)-\mathcal{E}_x''(t)\big)\mathrm{d}t\\
  &=\widetilde{u}(t)\frac{\cosh(1-x)\sinh t}{\sinh1}\Big|_0^x-\widetilde{u}(t)\frac{\cosh x\sinh(1-t)}{\sinh1}\Big|_x^1\\
  &=\frac{\widetilde{u}(x)}{\sinh1}\big(\cosh(1-x)\sinh x+\sinh(1-x)\cosh x\big)=\widetilde{u}(x)\\
  &=\Ev_x(u) \ ,
\end{align*}
where we have used $\displaystyle\mathcal{E}_x(t)-\mathcal{E}_x''(t)=0$ for $t\neq x$, $\mathcal{E}'_x|{(0,x)}\in H^1(0,x)$ and $\mathcal{E}'_x|{(x,1)}\in H^1(x,1)$.
As an interesting remark, it is worth noting that for $x,y\in[0,1]$ we have $\mathcal{E}_x(y)=\mathcal{E}_y(x)$.
\item[ii)] As a consequence of the previous result we immediately see that $\Ev_x$ is linear and continuous. This implies that $\Ev_x$ is differentiable and the differential is given by $\Ev_x$ itself.
\begin{align*}
\dd_u\Ev_x=\Ev_x\quad \mathrm{i.e.}\quad  \dd_u\Ev_x(h)=\Ev_x(h)=\tilde{h}(x)\,, \forall h\in H^1(0,1)\,.
\end{align*}
According to Definition \ref{def_functional_derivative} the functional derivative of $\Ev_x$ is then
\begin{align*}
D\Ev_x(u)=\mathcal{E}_x\,,
\end{align*}
which is independent of $u\in H^1(0,1)$.
\item[iii)] Obviously $\displaystyle \langle\mathcal{E}_x , \mathcal{E}_y\rangle_{H^1}=\mathcal{E}_x(y)=\mathcal{E}_y(x)$.
\item[iv)] It is not possible to define the evaluation of the derivative of an element of $u\in H^1(0,1)$ despite the fact that the continuous representative of such an $u$ is absolutely continuous and, hence, differentiable a.e. One way to convince oneself of the impossibility of such an endeavour is writing the Riesz-Fréchet representative of this map by introducing an orthonormal basis as in Theorem \ref{th_Riesz_Frechet} and checking that the result is not an element of $H^1(0,1)$. We discuss this in \ref{app_1}.
\end{enumerate}

\medskip

\noindent \textbf{2) An integral over $(0,1)$:} Let us study now the function
\begin{align*}
K:H^1(0,1)\rightarrow \mathbb{R}:u\mapsto \frac{1}{2}\int_{[0,1]}u^2\,.
\end{align*}
We first show that $K$ is differentiable. In view of
\begin{align*}
K(u+h)=\frac{1}{2}\int_{[0,1]}\big(u^2+2u h  +h^2   \big)
\end{align*}
it is natural to postulate
\begin{align*}
\dd_u K(h)=\int_{[0,1]}u h\,.
\end{align*}
This linear functional is continuous because
\begin{align*}
  \left|\int_{[0,1]}u h\right|&=\left|\int_{[0,1]}u h+u'h'-u'h' \right|=\big| \langle u,h\rangle_{H^1}-\langle u',h'\rangle_{L^2} \big| \\
  &\leq \big|\langle u,h\rangle_{H^1}\big|+\big|\langle u',h'\rangle_{L^2} \big| \leq \|u\|_{H^1} \|h\|_{H^1}+\|u'\|_{L^2} \|h'\|_{L^2}\\
  & \leq2\|u\|_{H^1}\|h\|_{H^1}\,,
\end{align*}
where we have used the Cauchy-Schwarz inequality and $\|\psi'\|_{L^2}\leq\|\psi\|_{H^1}$ for any $\psi\in H^1(0,1)$. In order to prove Fréchet-differentiability \eqref{def_differentiability} we compute
\begin{align}\label{limit1}
\begin{split}
\lim_{h\rightarrow0}\frac{1}{\|h\|}_{\!H^1} \! & \bigg|K(u+h)-K(u)-\int_{[0,1]}u h \bigg| \\
&= \lim_{h\rightarrow0}\frac{1}{2\|h\|}_{\!H^1}\!\bigg|\int_{[0,1]}h^2\bigg|\leq \lim_{h\rightarrow0}\frac{1}{2}\frac{\|h\|_{H^1}^2}{\|h\|_{H^1}}=0\,,
\end{split}
\end{align}
where we have made use of the continuity of the norm $\|\cdot\|_{H^1}$.

The functional derivative of $K$ can be neatly interpreted and understood by computing its evaluation at every $x\in[0,1]$
\begin{align*}
\Ev_x(DK(u))&=\langle \mathcal{E}_x , DK(u)\rangle_{H^1}=\dd_u K(\mathcal{E}_x)=\int_{[0,1]}u\mathcal{E}_x\\
&=\frac{\cosh(1-x)}{\sinh 1}\int_0^xu(t)\cosh t\,\mathrm{d}t+\frac{\cosh x}{\sinh 1}\int_x^1u(t)\cosh(1-t)\,\mathrm{d}t\,.
\end{align*}

Several comments are in order now:
\begin{itemize}
  \item $K(u)$ can be written as $\|u\|^2_{L^2}$.
  \item It is also possible to write $K$ in the form
  \begin{align*}
  K=\frac{1}{2}\int_0^1\Ev_x^2\,\mathrm{d}x\,,
  \end{align*}
  where the integral should be understood as a Bochner integral. This just means
  \begin{align*}
  K(u)=\frac{1}{2}\int_0^1\Ev_x^2(u)\,\mathrm{d}x=\frac{1}{2}\int_0^1\widetilde{u}^2(x)\,\mathrm{d}x=\frac{1}{2}\int_{[0,1]}u^2\,.
  \end{align*}
  \item $K$ is differentiable both in the Fréchet and Regge-Teitelboim senses.
\end{itemize}

\medskip

\noindent \textbf{3) A different type of integral over $(0,1)$:} Let us consider now the non-linear function
\begin{align*}
V:H^1(0,1)\rightarrow \mathbb{R}:u\mapsto\frac{1}{2}\int_{[0,1]}(u')^2\,.
\end{align*}
This is well defined because $u\in H^1(0,1)$.

This is an interesting example because such a function would be considered as non-differentiable (or \textit{non-admissible}) in the Regge-Teitelboim sense. Indeed, the standard computation gives
\begin{align*}
\delta V=\int_0^1 u'(\delta u)'=u'\delta u\big|_0^1-\int_0^1u''\delta u=u'(1)\delta u(1)-u'(0)\delta u(0)-\int_0^1 u''\delta u\,,
\end{align*}
which is not of the form
\begin{align*}
\delta V=\int_0^1\frac{\delta V}{\delta u}\delta u\,,
\end{align*}
owing to the presence of boundary terms. Notice, anyway, that the preceding computation would not be justified if $u''$ is not defined [which it does not have to for a generic $u\in H^1(0,1)$].

We show now that $V$ is Fréchet-differentiable. Indeed, as
  \begin{align*}
  V(u+h)=\frac{1}{2}\int_{[0,1]}\big((u')^2+2u'h'+(h')^2\big)\,,
  \end{align*}
  it is natural to postulate that
  \begin{align*}
  \dd_u V(h)=\int_{[0,1]}u'h'\,.
  \end{align*}
  In order to show that this is indeed the Fréchet-differential of $V$ we have to check its linearity and continuity and also
  \begin{equation}\label{limit}
\lim_{h\rightarrow0}\frac{1}{\|h\|}_{\!H^1}\!\bigg|V(u+h)-V(u)-\int_{[0,1]}u'h' \bigg|=  \lim_{h\rightarrow0}\frac{1}{2\|h\|}_{\!H^1}\!\bigg|\int_{[0,1]}(h')^2\bigg|=0\,.
  \end{equation}
  The linearity of $\dd_u V$ is obvious. Continuity is proven by
  \begin{align*}
  \bigg|\int_{[0,1]}u'h'\bigg|&= \bigg|\int_{[0,1]}\big(u'h'+u h-u h\big)\bigg|=\big|\langle u,h\rangle_{H^1}-\langle u,h\rangle_{L^2}\big| \\
  &\leq\big|\langle u,h\rangle_{H^1}\big|+\big|\langle u,h\rangle_{L^2}\big| \leq\|u\|_{H^1}\|h\|_{H^1}+\|u\|_{L^2}\|h\|_{L^2}\\
  &\leq2\|u\|_{H^1}\|h\|_{H^1}\,,
  \end{align*}
where we have used the Cauchy-Schwarz inequality and the fact that $\|\psi\|_{L^2}\leq\|\psi\|_{H^1}$ for every $\psi\in H^1(0,1)$. Finally
\begin{align*}
0\leq\frac{1}{\|h\|}_{\!H^1}\!\bigg|\int_{[0,1]}(h')^2\bigg|\leq\frac{1}{\|h\|}_{\!H^1}\!\bigg|\int_{[0,1]}\big(h^2+(h')^2\big)\bigg|=\|h\|_{H^1}\,,
\end{align*}
so that the continuity of the norm immediately gives equation \eqref{limit}.

A convenient description of the functional derivative $DV(u)\in H^1(0,1)$ can be obtained by considering its evaluation for $x\in(0,1)$
\begin{align*}
   \Ev_x&(DV(u))=\langle\mathcal{E}_x,DV(u)\rangle_{H^1}=\dd_u V(\mathcal{E}_x)=\int_{[0,1]}u'\mathcal{E}'_x\\
  &=\langle\mathcal{E}_x,u\rangle_{H^1}-\langle\mathcal{E}_x,u\rangle_{L^2}= \widetilde{u}(x)-\langle\mathcal{E}_x,u\rangle_{L^2}
  \\&=\widetilde{u}(x)-\frac{\cosh (1-x)}{\sinh 1}\int_0^xu(t)\cosh t\mathrm{d}t-\frac{\cosh x}{\sinh1}\int_x^1u(t)\cosh(1-t)\mathrm{d}t\,.
\end{align*}

\subsection{The \texorpdfstring{$H^1(\mathbb{R})$}{H1(R)} Sobolev space}

Before we discuss the evaluation in $H^1(\mathbb{R})$ we will prove that every element of this space has a continuous and \textit{bounded} representative.

\medskip

\begin{lemma}
If $u\in H^1(\mathbb{R})$ then there exists $C>0$ such that $|\widetilde{u}(x)|<C$ for all $x\in\mathbb{R}$.
\end{lemma}

\begin{proof}
We start by noting that, as a consequence of Theorem \ref{th_caclulus}, we have for every $y\in \mathbb{R}$
\begin{align*}
\widetilde{u}(x)=\widetilde{u}(x-y)+\int_0^yu'(x-\xi)\mathrm{d}\xi\,,
\end{align*}
which implies
\begin{align*}
|\widetilde{u}(x)|\leq |\widetilde{u}(x-y)|+\int_0^y|u'(x-\xi)|\mathrm{d}\xi\,.
\end{align*}
We average now in $y\in[0,1]$,
\begin{align*}
\int_0^1 |\widetilde{u}(x)|\mathrm{d}y\leq\int_0^1 |u(x-y)|\mathrm{d}y+\int_0^1 \int_0^y|u'(x-\xi)|\,\mathrm{d}\xi\, \mathrm{d}y\,,
\end{align*}
so that
\begin{align*}
|\widetilde{u}(x)|\leq \int_0^1 |u(x-y)|\mathrm{d}y+\int_0^1|u'(x-y)|\mathrm{d}y\,,
\end{align*}
where we have used the fact that for a fixed $x\in \mathbb{R}$ the maximum in $y\in[0,1]$ of the integral $\int_0^y|u'(x-\xi)|\,\mathrm{d}\xi$ is $\int_0^1|u'(x-\xi)|\,\mathrm{d}\xi$. Now, as $|u(x-y)|\leq\frac{1}{2}+\frac{1}{2}|u(x-y)|^2$ and $|u'(x-y)|\leq\frac{1}{2}+\frac{1}{2}|u'(x-y)|^2$ for every $x$ and $y$, we have for all $x\in\mathbb{R}$
\begin{align*}
|\widetilde{u}(x)|\leq1+\frac{1}{2}\int_0^1\big(|u(x-y)|^2+|u'(x-y)|^2\big)\mathrm{d}y\leq1+\frac{1}{2}\|u\|^2_{H^1(\mathbb{R})}<\infty\,,
\end{align*}
and, hence, we see that $\widetilde{u}$ is bounded.
\end{proof}

We study now the evaluation, that we denote $\Ev_x$ as before, and its Riesz-Fréchet representative $\mathcal{E}_x$. As we did above, we show that it is possible to write the evaluation as a scalar product in $H^1(\mathbb{R})$. In the present case we have
\begin{align*}
\mathcal{E}_x:\mathbb{R}\rightarrow \mathbb{R}:t\mapsto \frac{1}{2}e^{-|t-x|}\,.
\end{align*}
Indeed, it is easy to check that $\mathcal{E}_x\in H^1(\mathbb{R})$ because it is obviously an element of $L^2(\mathbb{R})$ and its weak derivative $\mathcal{E}'(t)=-\frac{1}{2}\mathrm{sgn}(t-x)e^{-|t-x|}$ is also an element of $L^2(\mathbb{R})$.

Let us take $u\in H^1(\mathbb{R})$, $x\in \mathbb{R}$ and compute
\begin{align*}
\langle\mathcal{E}_x,u\rangle_{H^1}&=\frac{1}{2}\int_{\mathbb{R}}e^{-|t-x|}\big(u(t)-\mathrm{sgn}(t-x)u'(t)\big)\mathrm{d}t\\
&=\frac{1}{2}\int_{-\infty}^x e^{t-x}\big(u(t)+u'(t)\big)+\frac{1}{2}\int_x^{\infty} e^{x-t}\big(u(t)-u'(t)\big)\\
&=\left.\frac{1}{2}\widetilde{u}(t)e^{t-x}\right|_{-\infty}^x-\left.\frac{1}{2}\widetilde{u}(t)e^{x-t}\right|_x^{\infty}\\
&=\widetilde{u}(x)-\frac{1}{2}\lim_{t\rightarrow-\infty}\widetilde{u}(t)e^{t-x}-\frac{1}{2}\lim_{t\rightarrow+\infty}\widetilde{u}(t)e^{x-t}=\widetilde{u}(x)\,,
\end{align*}
because $\widetilde{u}$ is bounded on the real line. As before, this proves that the evaluation is a linear and continuous function in $H^1(\mathbb{R})$. To end this section, we would like to draw the attention of the readers to the fact that, at variance with the first case considered in this section, some care is needed to work with Sobolev spaces defined on $\mathbb{R}$, in particular when performing integrations by parts.

\subsection{The \texorpdfstring{$\mathbbm{H}^2(\mathbb{R}^3)$}{H2(R3)} Sobolev space}

We discuss now the evaluation $\Ev^{(3)}_{\,\;\mathbf{x}}$ in the Sobolev space $\mathbbm{H}^2(\mathbb{R}^3)$ that we will use in our discussion of the holonomy-flux algebra. In the case of $H^1(I)$ or $H^1(\mathbb{R})$ we made use of the existence of continuous representatives for their elements. This allowed us to define the evaluation in terms of these representatives and prove its continuity. In the case of $H^2(\mathbb{R}^3)$ the second part of the Sobolev embedding theorem (see \cite{Adams}) tells us that $H^2(\mathbb{R}^3)$ is embedded in the H\"{o}lder space $C^{0,\lambda}(\mathbb{R}^3)$ (with $0<\lambda\leq1/2$), hence, also in $C^0(\mathbb{R}^3)$, so we have continuous representatives for the elements of $H^2(\mathbb{R}^3)$. The continuous representative of $u\in \mathbbm{H}^2(\mathbb{R}^3)$ will be denoted as $\widetilde{u}$. As the standard norm in $H^2(\mathbb{R}^3)$ is equivalent to the one that we are using for $\mathbbm{H}^2(\mathbb{R}^3)$ we can rely on the usual results, in particular the Sobolev embedding theorem, to work in $\mathbbm{H}^2(\mathbb{R}^3)$. Specifically, we will make use of the fact that $H^2(\mathbb{R}^3)\subset L^\infty(\mathbb{R}^3)$ and the inclusion is continuous. 

The main result of this subsection is the following

\medskip

\begin{proposition}
\label{11}
For a given $\mathbf{x_0} \in \mathbb{R}^3$ the function
\[
\mathcal{E}^{(3)}_{\;\mathbf{x_0}}(\mathbf{x}) = \frac{1}{8 \pi} e^{-\|\mathbf{x}-\mathbf{x_0}\|}
\]
is an element of $\mathbbm{H}^2(\mathbb{R}^3)$  and satisfies $\langle \mathcal{E}^{(3)}_{\;\mathbf{x_0}} , u \rangle_{\mathbbm{H}^2} = \widetilde{u}(\mathbf{x_0})$ for every $u\in\mathbbm{H}^2(\mathbb{R}^3)$ and $\mathbf{x_0}\in\mathbb{R}^3$.
\end{proposition}

\begin{proof}
In order to prove this claim notice that it suffices to consider $\mathbf{x_0} = 0$ because we can get the general result by a translation. In the following we will write $r :=\|\mathbf{x}\|_{\mathbb{R}^3}$ and make use of

\begin{align*}
&\mathcal{E}^{(3)}_{\;\,\mathrm{0}}(\mathbf{x}) = \frac{1}{8\pi} e^{-r}\,,\quad\partial_i \mathcal{E}^{(3)}_{\;\,\mathrm{0}}(\mathbf{x}) = -\frac{x_i}{r} \mathcal{E}^{(3)}_{\;\mathrm{0}}(\mathbf{x})\,,\\
&\partial_i \partial_j \mathcal{E}^{(3)}_{\;\,\mathrm{0}}(\mathbf{x}) = \left( \frac{1}{r^2} \Big( 1 + \frac{1}{r} \Big) x_i x_j - \frac{1}{r} \delta_{ij} \right) \mathcal{E}^{(3)}_{\mathrm{0}}(\mathbf{x})\,.
\end{align*}
It is straightforward to see that $\mathcal{E}^{(3)}_{\;\,\mathrm{0}}$, $\partial_i \mathcal{E}^{(3)}_{\;\,\mathrm{0}}$ and $\partial_i \partial_j \mathcal{E}^{(3)}_{\;\,\mathrm{0}}$ are in $L^2(\mathbb{R}^3)$, hence $\mathcal{E}^{(3)}_{\;\,\mathrm{0}}\in\mathbbm{H}^2(\mathbb{R}^3)$. In fact, we have
\begin{align*}
\|\mathcal{E}^{(3)}_{\;\,\mathrm{0}}\|_{\mathbbm{H}^2}^2=2\int_{\mathbb{R}^3}(\mathcal{E}^{(3)}_{\;\,\mathrm{0}})^2\Big(2+\frac{1}{r^2}\Big)=\frac{1}{8\pi}=\mathcal{E}^{(3)}_{\;\,\mathrm{0}}(0)\,.
\end{align*}
Let us take now a smooth function with compact support $\varphi\in C_c^\infty(\mathbb{R}^3)$ and compute the scalar product
\begin{align*}
\langle \mathcal{E}^{(3)}_{\;\,\mathrm{0}} ,\varphi\rangle_{\mathbbm{H}^2} &= \int_{\mathbb{R}^3} \mathcal{E}^{(3)}_{\;\,\mathrm{0}} \left(\varphi- 2  \sum_{i=1}^3 \frac{x_i}{r} \partial_i\varphi+  \sum_{i, j=1}^3 \left( \frac{1}{r^2} \left( 1 + \frac{1}{r} \right) x_i x_j - \frac{1}{r} \delta_{ij} \right) \partial_i \partial_j\varphi\right) \\
&=\int_{\mathbb{R}^3} \mathcal{E}^{(3)}_{\;\,\mathrm{0}} \left(\varphi- 2  \sum_{i=1}^3 \frac{x_i}{r} \partial_i\varphi+ \frac{1}{r^2}  \sum_{i, j=1}^3 x_ix_j\partial_i\partial_j \varphi+ \frac{1}{r^3} \sum_{i, j=1}^3 x_ix_j\partial_i\partial_j\varphi- \frac{1}{r}\Delta\varphi\right)\\
&=\int_{\mathbb{R}^3} \mathcal{E}^{(3)}_{\;\,\mathrm{0}} \left(\varphi-\frac{2}{r}\Big(1+\frac{1}{r^2}\Big) \sum_{i=1}^3x_i\partial_i \varphi+\frac{1}{r^2} \sum_{i, j=1}^3 x_ix_j\partial_i\partial_j \varphi\right)\\
&\ \ \ \ + \sum_{i, j=1}^3 \partial_i\left(\mathcal{E}^{(3)}_{\;\,\mathrm{0}}\frac{1}{r}\left(\frac{x_ix_j}{r^2}-\delta_{ij}\right)\partial_j \varphi\right)\,.
\end{align*}
where we have integrated by parts the last two terms in the integral of the second line. As $\varphi\in C_c^\infty(\mathbb{R}^3)$ the integral of the divergence term in the previous expression vanishes. Using spherical coordinates, $\partial_r\varphi= \sum_{i=1}^3\frac{x_i}{r} \partial_i \varphi$, and $\partial^2_r\varphi= \sum_{i,j=1}^3\frac{x_i x_j}{r^2} \partial_i \partial_j \varphi$ we get
\begin{align*}
\langle \mathcal{E}^{(3)}_{\;\,\mathrm{0}} ,\varphi\rangle_{\mathbbm{H}^2} = \int_{\mathbb{R}^3} \mathcal{E}^{(3)}_{\mathrm{0}} \left(\varphi- 2 \left( 1 + \frac{1}{r^2} \right) \partial_r\varphi+ \partial^2_r\varphi\right)\,.
\end{align*}
Using the expression for $\mathcal{E}^{(3)}_{\;\,\mathrm{0}}$ and denoting by $d\sigma$ the volume-form of  the unit sphere $S^2$, we find
\begin{align*}
\langle \mathcal{E}^{(3)}_{\;\,\mathrm{0}} ,\varphi\rangle_{\mathbbm{H}^2} &= \frac{1}{8\pi} \int_{S^2} d\sigma \int_0^\infty dr \ e^{-r} \left( r^2\varphi- 2 \left( 1 + r^2 \right) \partial_r\varphi+ r^2 \partial^2_r\varphi\right) =\\
&=  \frac{1}{8\pi} \int_{S^2} d\sigma \int_0^\infty dr\, \partial_r \Big(  r^2\partial_r(e^{-r}\varphi)-2\varphi e^{-r} (1 + r)\Big)\\
&= \frac{1}{8\pi} \int_{S^2} d\sigma \Big[r^2\partial_r(e^{-r}\varphi)-2\varphi e^{-r} (1 + r) \Big]^\infty_0 \\ 
&= \frac{1}{8\pi} 2 \varphi(0) \int_{S^2} d\sigma = \varphi(0)\,.
\end{align*}
In the previous computation we have taken advantage of the fact that $\varphi$ is smooth with compact support. 

To complete the proof (see details in \ref{app_2}) we have to show that the same result holds for any $u\in \mathbbm{H}^2(\mathbb{R}^3)$. To this end let us take a sequence of smooth functions with compact support $\varphi_n\in C_c^\infty(\mathbb{R}^3)$ converging to $u$ in $\mathbbm{H}^2(\mathbb{R}^3)$ (the existence of such a sequence is guaranteed because $C_c^\infty(\mathbb{R}^3)$ is dense in $\mathbbm{H}^2(\mathbb{R}^3)$).

According to the computation above, for every $\mathbf{x}\in\mathbb{R}^3$ we have
\[
\varphi_n(\mathbf{x})=\langle \mathcal{E}^{(3)}_{\;\mathbf{x}},\varphi_n\rangle_{\mathbbm{H}^2}
\]
which implies
\[
\lim_{n\rightarrow\infty}\varphi_n(\mathbf{x})=\lim_{n\rightarrow\infty}\langle \mathcal{E}^{(3)}_{\;\mathbf{x}},\varphi_n\rangle_{\mathbbm{H}^2}=\langle \mathcal{E}^{(3)}_{\;\mathbf{x}},u\rangle_{\mathbbm{H}^2}
\]
as  a consequence of the continuity of the scalar product.

The fact that $\varphi_n\in C_c^\infty(\mathbb{R}^3)$ converges to $u$ in $\mathbbm{H}^2(\mathbb{R}^3)$ and the continuous inclusion of $\mathbbm{H}^2(\mathbb{R}^3)$ into $L^\infty(\mathbb{R}^3)$ implies that the sequence also converges in $L^\infty(\mathbb{R}^3)$, i.e. for every $\varepsilon>0$ there exists $n_0\in\mathbb{N}$ such that $\|\varphi_n-u\|_\infty<\varepsilon$ if $n>n_0$. By taking the continuous representative $\widetilde{u}$ of $u$ we have that
\[
\|\varphi_n-u\|_\infty=\|\varphi_n-\widetilde{u}\|_\infty=\sup_{\mathbf{x}\in \mathbb{R}^3}|\varphi_n(\mathbf{x})-\widetilde{u}(\mathbf{x})|
\]
as a consequence of the fact that, for real continuous functions on $\mathbb{R}^3$, $\sup f=\mathrm{ess}\sup f$. This means that the sequence $\varphi_n$ converges \textit{uniformly} to $\widetilde{u}$ and, hence, also pointwise: $\lim_{n\rightarrow\infty}\varphi_n(\mathbf{x})=\widetilde{u}(\mathbf{x})$. We then conclude
\[
\widetilde{u}(\mathbf{x})=\langle \mathcal{E}^{(3)}_{\;\mathbf{x}},u\rangle_{\mathbbm{H}^2}\,.
\]

\end{proof}

An interesting observation is that it is not possible to directly guess the form of the representative of the evaluation $\Ev^{(3)}_{\,\;\mathbf{x}}$ in $\mathbbm{H}^2(\mathbb{R}^3)$ from that of the evaluation $\Ev_x$ in $H^1(\mathbb{R})$ (or in $H^1(\mathbb{R})$ for that matter), in particular $\mathcal{E}^{(3)}_{\mathbf{x}}$ cannot be simply written as $\left( \mathcal{E}_x\circ \pi_1\right)\cdot\left(\mathcal{E}_y\circ \pi_2\right)\cdot\left(\mathcal{E}_z\circ \pi_3\right)$ for $\mathbf{x}=(x,y,z)\in \mathbb{R}^3$ and $\pi_i : \mathbb{R}^3 \rightarrow \mathbb{R}$ being the projection to the corresponding factor.

%
%
\section{Poisson brackets in one dimensional examples}{\label{sec_one_dim_examples}}

In this section, we  discuss several sample computations of Poisson brackets in the phase space $H^1(0,1)\times H^1(0,1)^*\cong H^1(0,1)\times H^1(0,1)$. Our goal is to show that, despite what many statements found in the literature affirm, there is no problem in carrying out such computations, even if boundaries are present. In the following, whenever convenient, we will use the shorthand $H^1$ instead of $H^1(0,1)$.

Let us introduce the projections
\begin{align*}
  &\mathrm{proj}_1: H^1(0,1)\times H^1(0,1)\rightarrow  H^1(0,1):(\phi,\pi)\mapsto \phi\,,\\
  &\mathrm{proj}_2: H^1(0,1)\times H^1(0,1)\rightarrow  H^1(0,1):(\phi,\pi)\mapsto \pi\,,
\end{align*}
and
\begin{align*}
\Phi_x:=\Ev_x\circ\mathrm{proj}_1\,,\quad \Pi_x:=\Ev_x\circ\mathrm{proj}_2\,,
\end{align*}
with $x\in[0,1]$. These ``partial'' evaluations are real \textit{differentiable} functions in the phase space $H^1(0,1)\times H^1(0,1)$ because both projections $\mathrm{proj}_1$, $\mathrm{proj}_2$ and the evaluation $\Ev_x$ are differentiable. The definitions of $\Phi_x$ and $\Pi_x$ are unambiguous as a consequence of the uniqueness of the evaluation $\Ev_x$.

We compute now their Poisson bracket according to equation \eqref{PB_final}
\begin{align*}
\{\Phi_x,\Pi_y\}=\langle D_1\Phi_x,D_2 \Pi_y\rangle_{H^1}-\langle D_1\Pi_x,D_2 \Phi_y\rangle_{H^1}\,.
\end{align*}
To this end we need the relevant functional derivatives. In order to get $D_1\Phi_x$ we first compute
\begin{equation}\label{diff_1}
\dd_{(\phi,\pi)}\Phi_x=\dd_{(\phi,\pi)}(\Ev_x\circ\mathrm{proj}_1)=\dd_\phi\Ev_x\circ \dd_{(\phi,\pi)} \mathrm{proj}_1=\dd_\phi\Ev_x\circ\mathrm{proj}_1
\end{equation}
where we have made use of the chain rule and the fact that the projection $\mathrm{proj}_1$ is linear and continuous (hence, its own differential). For all $h\in H^1(0,1)$, and using Definition \ref{def_partial_functional_derivative} of partial functional derivative, we have
\begin{align*}
\dd_{(\phi,\pi)}\Phi_x(h,0)=\langle D_1 \Phi_x(\phi,\pi),h\rangle_{H^1}\,.
\end{align*}
Now, from equation \eqref{diff_1} we see that
\begin{align*}
\dd_{(\phi,\pi)}\Phi_x(h,0)=\dd_\phi\Ev_x(h)=\langle \mathcal{E}_x,h\rangle_{H^1}\,,
\end{align*}
and, hence,
\begin{align*}
D_1 \Phi_x(\phi,\pi)=\mathcal{E}_x\,.
\end{align*}
A similar computation gives $D_2 \Pi_y(\phi,\pi)=\mathcal{E}_y$. The computations of $D_2 \Phi_x$ and $D_1 \Pi_y$ are analogous. For instance,
\begin{align*}
\dd_{(\phi,\pi)}\Phi_x(0,h)=\langle D_2 \Phi_x(\phi,\pi),h\rangle_{H^1}\,.
\end{align*}
As
\begin{align*}
\dd_{(\phi,\pi)}\Phi_x(0,h)=\dd_\phi\Ev_x(\mathrm{proj}_1(0,h))= \dd_\phi\Ev_x(0)=0\,,
\end{align*}
we have $D_2\Phi_x(\phi,\pi)=0$ and, analogously, $D_1\Pi_x(\phi,\pi)=0$.

With all this information we then conclude
\begin{align}\label{basic_Poisson_brackets}
\{\Phi_x,\Pi_y\}&=\langle D_1\Phi_x,D_2 \Pi_y\rangle_{H^1}-\langle D_1\Pi_x,D_2 \Phi_y\rangle_{H^1}\\
&=\langle\mathcal{E}_x,\mathcal{E}_y\rangle_{H^1}=\mathcal{E}_x(y)=\mathcal{E}_y(x)\,.\nonumber
\end{align}
Several comments are in order now
\begin{itemize}
\item By using the functional derivatives computed above it is straightforward to see that $\{\Phi_x,\Phi_y\}=0$ and $\{\Pi_x,\Pi_y\}=0$. 
\item Both $\Phi_x$ and $\Pi_x$ are real functions in phase space. Their Poisson brackets are also real functions in phase space. Notice that, in the present example, these functions are \textit{constant} because they do not depend on the phase space point $(\phi,\pi)$.
\item It is very important to remember that the Poisson brackets \eqref{basic_Poisson_brackets} \textit{are completely determined by the canonical symplectic form} in $H^1(0,1)\times H^1(0,1)^*$. These can be interpreted as the basic Poisson brackets in this phase space and play the role of the $\{\phi(x), \pi(y)\}=\delta(x,y)$ in the usual presentations of the Hamiltonian formulation of field theories, with the replacement of $\delta(x,y)$ by $\mathcal{E}_x(y)$.
\item It is very important to bear in mind that in the example just discussed all the objects that we have used are suitably regular, in particular $\mathcal{E}_x$, which is an element of $H^1(0,1)$. Notice also that, at variance with the standard interpretation of the basic Poisson brackets for the scalar field, now $\{\Phi_x,\Pi_y\}$ is \textit{never} zero.
\end{itemize}

After computing the Poisson brackets of the \textit{canonical evaluations} $\Phi_x$ and $\Pi_y$ we compute now the Poisson brackets of these objects with the following real functions in the phase space $H^1(0,1)\times H^1(0,1)$
\begin{align}\label{KV_phase_space}
  & \mathcal{K}=K\circ \mathrm{proj}_1\,, \\
  & \mathcal{V}=V\circ \mathrm{proj}_1\,,
\end{align}
defined in terms of the functions $K$ and $V$ introduced in Section \ref{sec_sample_spaces}.

\medskip

We compute $\{\Phi_x,\mathcal{K}\}$ as
\begin{align*}
\{\Phi_x,\mathcal{K}\}=-\dd\mathcal{K}(\mathbb{X}_{\Phi_{\!x}}\!)\,,
\end{align*}
where $\mathbb{X}_{\Phi_{\!x}}\in \mathfrak{X}(H^1(0,1)\times H^1(0,1))$ is the Hamiltonian vector field determined by
\begin{align*}
\Omega(\mathbb{X}_{\Phi_{\!x}},\mathbb{Y})=\dd\Phi_x(\mathbb{Y})\,,\quad \forall \mathbb{Y}\in H^1(0,1)\times H^1(0,1)\,.
\end{align*}
In order to compute $\mathbb{X}_{\Phi_{\!x}}$ we write $\mathbb{X}_{\Phi_x}\!\!=(X_1,X_2)$, $\mathbb{Y}=(Y_1,Y_2)$. Now,
\begin{align*}
  & \Omega(\mathbb{X}_{\Phi_{\!x}},\mathbb{Y})\!=\langle Y_2,X_1 \rangle_{H^1}-\langle X_2,Y_1 \rangle_{H^1}\,,\phantom{\Big(} \\
  & \dd\Phi_x(\mathbb{Y})=(\dd\Ev_x\circ\mathrm{proj}_1)(\mathbb{Y})=\dd\Ev_x(Y_1)=\langle\mathcal{E}_x,Y_1\rangle_{H^1}\,,
\end{align*}
hence, $\mathbb{X}_{\Phi_{\!x}}\!\!=(0,-\mathcal{E}_x)$, and, remembering that the Poisson bracket is a real function in phase space,
\begin{align*}
\{\Phi_x,\mathcal{K}\}(\phi,\pi)&=-\dd_{(\phi,\pi)}\mathcal{K}(\mathbb{X}_{\Phi_{\!x}})=-\big(\dd_\phi K\circ \mathrm{proj}_1\big)(\mathbb{X}_{\Phi_{\!x}}\!)\\
&=-\dd_\phi K(X_1)=-\dd_\phi K(0)=0\,.
\end{align*}
We then conclude $\{\Phi_x,\mathcal{K}\}=0$.

\medskip

The computation of $\,\{\Pi_y,\mathcal{K}\}$ is analogous. First we write
\begin{align*}
\{\Pi_y,\mathcal{K}\}=-\dd\mathcal{K}(\mathbb{X}_{\Pi_{\!y}}\!)\,,
\end{align*}
where $\mathbb{X}_{\Pi_{\!y}}$ is the Hamiltonian vector field defined by $\Pi_y$. Writing as above $\mathbb{X}_{\Pi_y}\!\!=(X_1,X_2)$ and $\mathbb{Y}=(Y_1,Y_2)$ we have
\begin{align*}
  & \Omega(\mathbb{X}_{\Pi_{\!y}},\mathbb{Y})\!=\langle Y_2,X_1 \rangle_{H^1}-\langle X_2,Y_1 \rangle_{H^1}\,,\phantom{\Big(} \\
  & \dd\Pi_y(\mathbb{Y})=(\dd\Ev_y\circ\mathrm{proj}_2)(\mathbb{Y})=\dd\Ev_y(Y_2)=\langle\mathcal{E}_y,Y_2\rangle_{H^1}\,,
\end{align*}
hence, $\mathbb{X}_{\Pi_{\!y}}\!\!=(\mathcal{E}_y,0)$. Now
\begin{align*}
\{\Pi_y,\mathcal{K}\}(\phi,\pi)&=-\dd_{(\phi,\pi)}\mathcal{K}(\mathbb{X}_{\Pi_{\!y}}\!)=-\big(\dd_\phi K\,\circ\, \mathrm{proj}_1\big)(\mathbb{X}_{\Pi_{\!y}})\\
&=-\dd_\phi K(X_1)=-\dd_\phi K(\mathcal{E}_y)=-\int_{[0,1]}\!\!\mathcal{E}_y\phi\,.
\end{align*}
This can also be written in the form
\begin{align*}
\{\Pi_y,\mathcal{K}\}=-\int_0^1\!\!\mathcal{E}_y(x)\Phi_x\, \mathrm{d}x\,.
\end{align*}
Before considering other examples we make some comments:
\begin{itemize}
  \item It is straightforward to prove that the Poisson brackets computed above are real Fréchet-differentiable functions in the phase space $H^1(0,1)\times H^1(0,1)$.
  \item As expected, there is a direct and streamlined way to perform the previous computations:
  \begin{align}\label{Poisson_comp_streamlined}
    \{\Phi_x,\mathcal{K}\} & =\{\Phi_x,\frac{1}{2}\int_0^1\!\!\Phi_y^2\mathrm{d}y\}=\int_0^1\!\!\Phi_y\{\Phi_x,\Phi_y\}\,\mathrm{d}y=0\,,   \\
    \{\Pi_x,\mathcal{K}\} & =\{\Pi_x,\frac{1}{2}\int_0^1\!\!\Phi_y^2\mathrm{d}y\}=\int_0^1\!\!\Phi_y\{\Pi_x,\Phi_y\}\,\mathrm{d}x=-\int_0^1\!\!\mathcal{E}_x(y)\Phi_y \mathrm{d}y\,,
  \end{align}
  where we have used the well-known properties of the Poisson bracket. The reason why we have followed a slightly longer path above is to highlight how the computations can be performed in terms of Fréchet differentials and the role played by the scalar product in $H^1(0,1)$.
  \item The fact that our fields and momenta are defined in a manifold with boundary (the interval $(0,1)$) has very little impact in the previous computations: nothing strange happens at the points $0,1$.
\end{itemize}

\medskip

In the following we will perform similar computations with the function $\mathcal{V}$ which is not RT-differentiable. As we will see, there are no obstructions to get the Poisson brackets despite the purported lack of regularity and the presence of boundaries.

\medskip

We start with $\{\Phi_x,\mathcal{V}\}=-\dd\mathcal{V}(\mathbb{X}_{\Phi_{\!x}})$. As before we have $\mathbb{X}_{\Phi_{\!x}}=(0,-\mathcal{E}_x)$ and, hence,
\begin{align*}
\{\Phi_x,\mathcal{V}\}(\phi,\pi)&=-\dd_{(\phi,\pi)}\mathcal{V}(\mathbb{X}_{\Phi_{\!x}}\!)=-\big(\dd_\phi V\circ \mathrm{proj}_1\big)(\mathbb{X}_{\Phi_{\!x}}\!)\\
&=-\dd_\phi V(X_1)=-\dd_\phi V(0)=0\,,\qquad
\end{align*}
i.e. $\{\Phi_x,\mathcal{V}\}=0$.

\medskip

The computation of $\{\Pi_y,\mathcal{V}\}=-\dd\mathcal{V}(\mathbb{X}_{\Pi_{\!y}})$ is carried out in a completely analogous way by taking into account that we have $\mathbb{X}_{\Pi_{\!y}}=(\mathcal{E}_y,0)$ and
\begin{align*}
\{\Pi_y,\mathcal{V}\}(\phi,\pi)\!&=\!-\dd_{(\phi,\pi)}\mathcal{V}(\mathbb{X}_{\Pi_{\!y}}\!)\!=\!-\big(\dd_\phi V\,\circ\, \mathrm{proj}_1\big)(\mathbb{X}_{\Pi_{\!y}}\!)\!=\!-\dd_\phi V(X_1) \\
\!&=\!-\dd_\phi V(\mathcal{E}_y)\!=\!-\int_{[0,1]}\!\!\!\mathcal{E}_y'\phi'\,.
\end{align*}
Several comments are in order now
\begin{itemize}
  \item There are no obstructions to perform the previous computations, in particular, the presence of the boundary does not introduce any complications.
  \item The resulting Poisson brackets are Fréchet-differentiable functions. In particular, $\{\Pi_y,\mathcal{V}\}$ is differentiable because it can be written as $F=f\circ \mathrm{proj}_1$ with
  \begin{align*}
  f:H^1(0,1)\rightarrow \mathbb{R}:u\mapsto -\int_{[0,1]}\!\!\!\mathcal{E}_y'u'\,,
  \end{align*}
linear and continuous because
\begin{align*}
|f(u)|&=\bigg|\int_{[0,1]}\mathcal{E}_y'u'\bigg|=\big|\langle\mathcal{E}_y,u\rangle_{H^1}-\langle\mathcal{E}_y,u\rangle_{L^2} \big|\leq\big|\langle\mathcal{E}_y,u\rangle_{H^1}\big|+\big|\langle\mathcal{E}_y,u\rangle_{L^2}\big|\\
&\leq\|\mathcal{E}_y\|_{H^1}\|u\|_{H^1}+\|\mathcal{E}_y\|_{L^2}\|u\|_{L^2}\leq2\|\mathcal{E}_y\|_{H^1}\|u\|_{H^1}\,.
\end{align*}
  \item The impossibility of defining the derivative of the evaluation $\Phi_x'$ as a continuous linear functional in $H^1(0,1)$ (see \ref{app_1}) precludes us from writing an expression such as
  \begin{align*}
  \mathcal{V}=\frac{1}{2}\int_0^1 (\Phi_x')^2\mathrm{d}x\,,
  \end{align*}
 and use it to compute Poisson brackets.
\end{itemize}

%
%
\section{The mock holonomy-flux algebra}{\label{sec_mock_HFalgebra}}

The purpose of this section is to introduce a phase space modelled on Sobolev spaces of the types introduced above where it is possible to define holonomy and flux variables which mimic those used in LQG. The main goal of this exercise is to show that, with some care, and relying on the structures provided by the functional spaces where the basic variables are defined, it is possible to study their algebra without encountering any surprises (for instance of the type found in \cite{ACZ} when looking at the Jacobi identity to justify the necessity of having non-commuting fluxes).

Let us introduce the phase space $\mathbbm{H}^2(\mathbb{R}^3)^9\times \mathbbm{H}^2(\mathbb{R}^3)^{9*}\cong \mathbbm{H}^2(\mathbb{R}^3)^9\times \mathbbm{H}^2(\mathbb{R}^3)^9$ with elements that we will write as $(A_a^i,\hat{E}^b_j)$. The indices $a,b=1,2,3$ refer to a Cartesian coordinate system in $\mathbb{R}^3$, $i,j=1,2,3$ are internal indices. These variables will play the role of an $SO(3)$ connection and its canonically conjugate momentum on the spatial manifold $\mathbb{R}^3$.

Holonomies and fluxes will be defined with the help of the evaluation $\Ev^{(3)}_{\;\,\mathbf{x}}$ in $\mathbbm{H}^2(\mathbb{R}^3)$, discussed in Section \ref{sec_sample_spaces}, and the projections onto the different factors in each copy of $\mathbbm{H}^2(\mathbb{R}^3)^9$
\begin{align*}
  \mathrm{proj}_{Aa}^{\phantom{A}\,i}&:\mathbbm{H}^2(\mathbb{R}^3)^9\times \mathbbm{H}^2(\mathbb{R}^3)^9\rightarrow \mathbbm{H}^2(\mathbb{R}^3):(A,\hat{E})\mapsto A_a^i\,,\\
  \mathrm{proj}_{E\,i}^{\phantom{E}a}&:\mathbbm{H}^2(\mathbb{R}^3)^9\times \mathbbm{H}^2(\mathbb{R}^3)^9\rightarrow \mathbbm{H}^2(\mathbb{R}^3):(A,\hat{E})\mapsto \hat{E}^a_i\,,
\end{align*}
where $A_a^i$ and $\hat{E}^a_i$ denote the components of the connection $A$ and the triad $\hat{E}$.

For $\mathbf{x}\in \mathbb{R}^3$ let us define now
\begin{align*}
  \mathcal{A}_{\mathbf{x}a}^{\phantom{x}\,i}&:=\Ev^{(3)}_{\;\,\mathbf{x}}\circ  \mathrm{proj}_{Aa}^{\phantom{A}\,i}\,,\\
  \hat{\mathfrak{E}}_{\mathbf{x}\,i}^{\phantom{x}a}&:=\Ev^{(3)}_{\;\,\mathbf{x}}\circ  \mathrm{proj}_{E\,i}^{\phantom{E}a}\,.
\end{align*}
The basic Poisson brackets of $ \mathcal{A}_{\mathbf{x}a}^{\phantom{x}\,i}$ and $\hat{\mathfrak{E}}_{\mathbf{x}\,i}^{\phantom{x}a}$ are
\begin{align}
  \{\mathcal{A}_{\mathbf{x}a}^{\phantom{x}\,i},\mathcal{A}_{\mathbf{y}b}^{\phantom{y}\,j}\} &=0\,, \label{basic_PB_3dim_AA}\\
  \{\hat{\mathfrak{E}}_{\mathbf{x}\,i}^{\phantom{x}a}\,\,,\,\hat{\mathfrak{E}}_{\mathbf{y}\,j}^{\phantom{y}b}\} &=0\,, \label{basic_PB_3dim_EE}\\
  \{\mathcal{A}_{\mathbf{x}a}^{\phantom{x}\,i},\,\hat{\mathfrak{E}}_{\mathbf{y}\,j}^{\phantom{y}b}\} &=\delta_a^b\delta_j^i\,\mathcal{E}^{(3)}_{\;\,\mathrm{\mathbf{x}}}(\mathbf{y})\,.\label{basic_PB_3dim_AE}
\end{align}
These Poisson brackets can be easily computed by using the approach described in Section \ref{sec_one_dim_examples}.

\medskip

In order to define the flux variables we fix a smooth compact surface $S$ embedded in $\mathbb{R}^3$ and test fields $f^i\in \mathbbm{H}^2(\mathbb{R}^3)$. We then introduce the functional
\begin{equation}\label{flux}
^2E_S[f]:\mathbbm{H}^2(\mathbb{R}^3)^9\times \mathbbm{H}^2(\mathbb{R}^3)^9\rightarrow \mathbb{R}:(A,\hat{E})\mapsto \int_S f^i\,E_i\,,
\end{equation}
where the 2-forms $E_i$ are such that the vector fields $\hat{E}^a_i(\mathbf{x})=\hat{\mathfrak{E}}_{\mathbf{x}\,i}^{\phantom{x}a}(A,\hat{E})$ are defined by the dual objects
\begin{align*}
\left(\frac{\cdot\wedge E_i}{v}\right)
\end{align*}
(see \ref{app_3} for details). 

In analogy with the one dimensional examples discussed in Section \ref{sec_one_dim_examples} we write the phase space functions $^2E_S[f]$ in terms of the evaluations $\hat{\mathfrak{E}}_{\mathbf{x}\,i}^{\phantom{x}a}$
\begin{equation}\label{flux2}
^2E_S[f]=\int_S f^i\, \imath_{\hat{\mathfrak{E}}_i}\!v\,,
\end{equation}
In order to conform with the standard notations and facilitate the comparison of our computations with the standard ones we will write \eqref{flux2} in the form
\begin{equation}\label{flux_3}
^2E_S[f]=\int_{\mathbf{x}\in S} \mathrm{d} S^{bc}(\mathbf{x})f^i(\mathbf{x})\hat{\mathfrak{E}}_{\mathbf{x}i}^a v_{abc}(\mathbf{x})\,,
\end{equation}
where, as commented in \ref{app_3}, the volume form $v$ is independent of the canonical variables. 

The holonomies associated with a piecewise smooth curve $\gamma:[0,1]\rightarrow \mathbb{R}^3$ can also be defined in terms of evaluations in the usual way involving a path-ordered exponential:
\begin{equation}\label{holonomy1}
  h_\gamma:\mathbbm{H}^2(\mathbb{R}^3)^9\times \mathbbm{H}^2(\mathbb{R}^3)^9\rightarrow \mathbb{R}:(A,\hat{E})\mapsto \mathcal{P}\exp \left(\int_\gamma \mathcal{A}\right)\,.
\end{equation}
The holonomies defined on the restriction of $\gamma$ to the interval $(\lambda_1,\lambda_2)$ with $0\leq \lambda_1<\lambda_2\leq 1$ will be denoted as $h_\gamma(\lambda_2,\lambda_1)$. They can be written explicitly in the form
\begin{align}
\nonumber
h_\gamma(\lambda_2,\lambda_1)&=\mathbbm{1} +\!\sum_{n=1}^\infty\! \int_{\lambda_1}^{\lambda_2}\!\!\!\!\mathrm{d}t_1\!\int_{\lambda_1}^{t_1}\!\!\!\!\!\mathrm{d}t_2\,\sdots\int_{\lambda_1}^{t_{n-1}}\!\!\!\!\!\!\!\!\!\!\mathrm{d}t_n  \dot{\gamma}^{a_1}(t_1)\mathcal{A}_{\bm{\gamma}(t_1\!)\,a_1}^{\phantom{\bm{\gamma}(t_1)}\!\!i_1}\! \frac{\tau_{i_1}}{2} \sdots \dot{\gamma}^{a_n}(t_n)\mathcal{A}_{\bm{\gamma}(t_n\!)\,a_n}^{\phantom{\bm{\gamma}(t_n)}\!\!i_n}\!\frac{\tau_{i_n}}{2} \\\label{holonomy2}
&= \mathbbm{1} +\!\sum_{n=1}^\infty\! \int_{\lambda_1}^{\lambda_2}\!\!\!\!\mathrm{d}t_1\!\int_{\lambda_1}^{t_1}\!\!\!\!\!\mathrm{d}t_2\,\sdots\int_{\lambda_1}^{t_{n-1}}\!\!\!\!\!\!\!\!\!\!\mathrm{d}t_n  \, A_\gamma \left( t_1 \right)  \sdots A_\gamma \left(  t_n \right)
\end{align}
where $\tau_j=  \mathrm{i} \sigma_j$, the $\sigma_j$ are the Pauli matrices, $t_0=1$, and 
\[A_\gamma \left( t_j \right):= \dot{\gamma}^{a_j}(t_j) \mathcal{A}_{\bm{\gamma}(t_j\!)\,a_j}^{\phantom{\bm{\gamma}(t_j)}\!\!i_j}\! \frac{\tau_{i_j}}{2} \,.\]  
It is possible to prove that the Dyson series \eqref{holonomy2} converges \cite{Baez}.  The variables used in LQG are traces of holonomies along closed curves $\mathrm{Tr}\,h_\gamma(1,0)=:\mathrm{Tr}\,h_\gamma$. 

As the previous expressions are written in terms of the evaluations $\hat{\mathfrak{E}}_{\mathbf{x}\,i}^{\phantom{x}a}$ and $\mathcal{A}_{\mathbf{x}a}^{\phantom{x}\,i}$, the computation of Poisson brackets involving these variables is straightforward and can be carried out by following essentially the same steps as in the standard computations \cite{Thiemann1} (for other approaches see \cite{Lew1,GROSS19851}). Notice in particular that we will always have
\begin{align*}
\big\{^2E_{S_1}[f_1],^2E_{S_2}[f_2]\big\}=0\,,\quad \big\{\mathrm{Tr}\,h_{\gamma_1},\mathrm{Tr}\,h_{\gamma_2} \big\}=0\,,
\end{align*}
for smooth surfaces $S_1,S_2$, test functions $f_1,f_2$, and piecewise smooth curves $\gamma_1,\gamma_2$.
\begin{proposition}\label{proppbhf}
The Poisson bracket of an holonomy and a flux is given by
\begin{align*}
\big\{h_\gamma(\lambda_2,\lambda_1),^2\!E_S[f] \big\}\!=\!\!\int_{x\in S}\!\!\!\!\!\!\mathrm{d} S^{bc}(\mathbf{x})f^k(\mathbf{x})v_{abc}(\mathbf{x})\!\int_{\lambda_1}^{\lambda_2}\!\!\!\mathrm{d}t\dot{\gamma}^a(t) \mathcal{E}^{(3)}_{\;\,\gamma(t)}(\mathbf{x})h_\gamma(\lambda_2,t)\frac{\tau_k}{2}h_\gamma(t,\lambda_1 )\,.
\end{align*}
\end{proposition}
\begin{proof}

To prove this result we use the definitions of the holonomy \eqref{holonomy2}, of the flux \eqref{flux_3}, and the basic Poisson brackets \eqref{basic_PB_3dim_AE}
\begin{align*}
  &\big\{h_\gamma(\lambda_2,\lambda_1),^2\!E_S[f] \big\}\\
  &= \int_{\lambda_1}^{\lambda_2}\mathrm{d}t_1\{A_\gamma(t_1),^2E_S[f]\}\\
  &+ \sum_{n=2}^{\infty}\int_{\lambda_1}^{\lambda_2}\mathrm{d}t_1\int_{\lambda_1}^{t_1}\mathrm{d}t_2\cdots\int_{\lambda_1}^{t_{n-1}}\mathrm{d}t_n\{A_\gamma(t_1),^2E_S[f]\}A_\gamma(t_2)\cdots A_\gamma(t_n)\\
  &+\sum_{n=2}^{\infty}\int_{\lambda_1}^{\lambda_2}\mathrm{d}t_1\int_{\lambda_1}^{t_1}\mathrm{d}t_2\cdots\int_{\lambda_1}^{t_{n-1}}\mathrm{d}t_n A_\gamma(t_1)\cdots A_\gamma(t_{n-1})\{A_\gamma(t_n),^2E_S[f]\}\\
  &+ \sum_{n=3}^{\infty}\sum_{k=2}^{n-1}\int_{\lambda_1}^{\lambda_2}\mathrm{d}t_1\int_{\lambda_1}^{t_1}\mathrm{d}t_2\cdots\int_{\lambda_1}^{t_{n-1}}\mathrm{d}t_n
  A_\gamma(t_1)\cdots  A_\gamma(t_{k-1})\\
  &\hspace*{7.2cm}\times\{A_\gamma(t_k),^2E_S[f]\}A_\gamma(t_{k+1})\cdots A_\gamma(t_n)\,,
\end{align*}
where
\[
\{A_\gamma(t),^2 E_S[f]\}=\frac{\tau_k}{2}\,\dot{\gamma}^a(t)\int_{\mathbf{x} \in S}\mathrm{d} S^{bc}(\mathbf{x})f^k(\mathbf{x})\mathcal{E}^{(3)}_{\gamma(t)}(\mathbf{x})v_{abc}(\mathbf{x})\,.
\]
As it can be seen, we have split the computation into several pieces because it is convenient to separately deal with the terms in which the Poisson brackets $\{A_\gamma(t_k),^2E_S[f]\}$ appear either at the beginning or the end from those where this Poisson bracket is embedded in an expression of the form $A_\gamma(t_1)\cdots  A_\gamma(t_{k-1})\{A_\gamma(t_k),^2E_S[f]\}A_\gamma(t_{k+1})\cdots A_\gamma(t_n)$.

By using
\[
\int_{\lambda_1}^{\lambda_2}\!\!\!\mathrm{d}t_1\int_{\lambda_1}^{t_1}\!\!\!\mathrm{d}t_2\sdots \int_{\lambda_1}^{t_{k-1}}\!\!\!\!\!\!\!\!\mathrm{d}t_k\int_{\lambda_1}^{t_k}\!\!\!\mathrm{d}t \,f(t_1\ldots,t_k,t)=
\int_{\lambda_1}^{\lambda_2}\!\!\!\mathrm{d}t\int_{t}^{\lambda_2}\!\!\!\mathrm{d}t_1\int_{t}^{t_1}\!\!\!\mathrm{d}t_2\sdots \int_{t}^{t_{k-1}}\!\!\!\!\!\!\!\!\mathrm{d}t_kf(t_1\ldots,t_k,t)
\]
(which can be easily proven by induction and obtained by interchanging contiguous integrals two at a time), renaming integration variables as needed, and employing some elementary identities for double sums, it is straightforward to arrive at the desired result.
\end{proof}

Using the proposition \ref{proppbhf}, we have the following 
\begin{coro}
The Poisson bracket of the trace of an holonomy and a flux is given by
\begin{align*}
\big\{\mathrm{Tr}\,h_{\gamma}, ^2\!E_S[f]  \big\}&\!=\!\int_{\mathbf{x}\in S} \!\!\!\!\!\!\mathrm{d} S^{bc}(\mathbf{x}) \, f^i(\mathbf{x}) v_{abc}(\mathbf{x}) \int_0^1\!\!\!\mathrm{d} t \, \dot{\gamma}^{a} (t) \mathcal{E}^{(3)}_{\;\,\gamma(t)}(\mathbf{x})  \mathrm{Tr}\,\Big(h_{\gamma} (1, t) \frac{\tau_i}{2}  h_{\gamma} (t, 0)\Big) \,.
\end{align*}
\end{coro}

By using the results obtained so far we are now in the position to explicitly check that the Jacobi identity is satisfied. Let us write
\begin{align}
\label{jacobi}
\begin{split}
      \mathfrak{J} := \Big\{  \big\{ \mathrm{Tr} \,h_{\gamma}, \ ^2E_S[f]  \big\}, \ ^2E_S[g]\Big\}&+\Big\{  \big\{^2E_S[f], ^2E_S[g]   \big\}, \mathrm{Tr} \, h_{\gamma} \Big\} \\
  &+  \Big\{ \big\{^2E_S[g],\mathrm{Tr} \, h_{\gamma}  \big\}, ^2E_S[f] \Big\} \ .
\end{split}
\end{align}
It will be enough to compute just one of the terms
\begin{align*}
    \Big\{  &\big\{ \mathrm{Tr} \, h_{\gamma}, \ ^2E_S[f]  \big\}, \ ^2E_S[g]\Big\} \\
    &= \int_{\mathbf{x}\in S} \!\!\!\mathrm{d} S^{bc}(\mathbf{x}) \, f^i(\mathbf{x}) v_{abc}(\mathbf{x})\int_0^1\!\!\!\mathrm{d} t \, \dot{\gamma}^{a} (t) \mathcal{E}^{(3)}_{\;\,\gamma(t)}(\mathbf{x})  \mathrm{Tr}\big(\big\{ h_{\gamma} (1, t),  ^2E_S[g] \big\}  \frac{\tau_i}{2}  h_{\gamma} (t, 0)\big) \\
    &+ \int_{\mathbf{x}\in S} \!\!\!\mathrm{d} S^{bc}(\mathbf{x}) \, f^i(\mathbf{x}) v_{abc}(\mathbf{x})\int_0^1\!\!\!\mathrm{d} t \, \dot{\gamma}^{a} (t) \mathcal{E}^{(3)}_{\;\,\gamma(t)}(\mathbf{x}) \mathrm{Tr}\,\big(h_{\gamma} (1, t)  \frac{\tau_i}{2} \big\{ h_{\gamma} (t, 0) ,  ^2E_S[g] \big\}\big) \\
    &= \int_{\mathbf{y}\in S} \!\!\!\mathrm{d} S^{ef}(\mathbf{y}) g^j(\mathbf{y})v_{def}(\mathbf{y}) \int_{\mathbf{x}\in S} \!\!\!\mathrm{d} S^{bc}(\mathbf{x}) \, f^i(\mathbf{x}) v_{abc}(\mathbf{x}) \\
    & \hspace{30bp} \times \int_0^1\!\!\!\mathrm{d} t \int_t^1\!\!\!\mathrm{d} s \, \dot{\gamma}^{a}(t) \dot{\gamma}^{d} (s) \mathcal{E}^{(3)}_{\;\,\gamma(t)}(\mathbf{x}) \mathcal{E}^{(3)}_{\;\,\gamma(s)}(\mathbf{y})   \mathrm{Tr} \big( h_{\gamma} (1, s) \frac{\tau_j}{2}  h_{\gamma} (s,t) \frac{\tau_i}{2}  h_{\gamma} (t, 0) \big)\\
    &+\int_{\mathbf{y}\in S} \!\!\!\mathrm{d} S^{ef}(\mathbf{y}) g^j(\mathbf{y})v_{def}(\mathbf{y}) \int_{\mathbf{x}\in S} \!\!\!\mathrm{d} S^{bc}(\mathbf{x}) \, f^i(\mathbf{x}) v_{abc}(\mathbf{x})  \\
    &\hspace{30bp} \times \int_0^1\!\!\!\mathrm{d} t \int_0^t\!\!\!\mathrm{d} s \, \dot{\gamma}^{a}(t) \dot{\gamma}^{d} (s) \mathcal{E}^{(3)}_{\;\,\gamma(t)}(\mathbf{x}) \mathcal{E}^{(3)}_{\;\,\gamma(r)}(\mathbf{y})  \mathrm{Tr} \big( h_{\gamma} (1, t) \frac{\tau_i}{2}  h_{\gamma} (t,s) \frac{\tau_j}{2}  h_{\gamma} (s, 0) \big)\, .
\end{align*}

Note that
\begin{align}
    \int_0^1 \mathrm{d} t  \int_t^1 \mathrm{d} s \,  F(t,s) = \int_0^1 \mathrm{d}s \int_0^s \mathrm{d}t \, F(t,s) \ , \label{fubini1}\\
    \int_0^1 \mathrm{d}t \int_0^t \mathrm{d}s \, F(t,s) = \int_0^1 \mathrm{d}s \int_s^1 \mathrm{d}t  \, F(t,s)\ .\label{fubini2}
\end{align}
By renaming $i \leftrightarrow j$, $\mathbf{x} \leftrightarrow \mathbf{y}$, $t \leftrightarrow s$, $a\leftrightarrow d$, $b\leftrightarrow e$ and $c\leftrightarrow f$ and using equations \eqref{fubini1} and \eqref{fubini2} one obtains

\begin{align*}
    \Big\{  \big\{& \mathrm{Tr} \, h_{\gamma}, \ ^2E_S[f]  \big\}, \ ^2E_S[g]\Big\} \\
    &= \int_{\mathbf{y}\in S} \!\!\!\mathrm{d} S^{ef}(\mathbf{y}) f^j(\mathbf{y})v_{def}(\mathbf{y}) \int_{\mathbf{x}\in S} \!\!\!\mathrm{d} S^{bc}(\mathbf{x}) \, g^i(\mathbf{x}) v_{abc}(\mathbf{x}) \\
    & \hspace{30bp} \times \int_0^1\!\!\!\mathrm{d} t \int_t^1\!\!\!\mathrm{d} r \, \dot{\gamma}^{a}(t) \dot{\gamma}^{d} (r) \mathcal{E}^{(3)}_{\;\,\gamma(t)}(\mathbf{x}) \mathcal{E}^{(3)}_{\;\,\gamma(r)}(\mathbf{y})   \mathrm{Tr} \big( h_{\gamma} (1, r) \frac{\tau_j}{2}  h_{\gamma} (r,t) \frac{\tau_i}{2}  h_{\gamma} (t, 0) \big)\\
    &+\int_{\mathbf{y}\in S} \!\!\!\mathrm{d} S^{ef}(\mathbf{y}) f^j(\mathbf{y})v_{def}(\mathbf{y}) \int_{\mathbf{x}\in S} \!\!\!\mathrm{d} S^{bc}(\mathbf{x}) \, g^i(\mathbf{x}) v_{abc}(\mathbf{x})  \\
    &\hspace{30bp} \times \int_0^1\!\!\!\mathrm{d} t \int_0^t\!\!\!\mathrm{d} r \, \dot{\gamma}^{a}(t) \dot{\gamma}^{d} (r) \mathcal{E}^{(3)}_{\;\,\gamma(t)}(\mathbf{x}) \mathcal{E}^{(3)}_{\;\,\gamma(r)}(\mathbf{y})  \mathrm{Tr} \big( h_{\gamma} (1, t) \frac{\tau_i}{2}  h_{\gamma} (t,r) \frac{\tau_j}{2}  h_{\gamma} (r, 0) \big)\\
    = \Big\{  \big\{& \mathrm{Tr} \, h_{\gamma},\ ^2E_S[g]  \big\}, \ ^2E_S[f]\Big\} \ .
\end{align*}

Since $\Big\{  \big\{^2E_S[f], ^2E_S[g]   \big\}, \mathrm{Tr} \, h_{\gamma} \Big\}  = 0$, equation \eqref{jacobi} reduces to
\begin{align*}
    \mathfrak{J} &= \Big\{  \big\{\mathrm{Tr} \, h_{\gamma}, \ ^2E_S[f]  \big\}, \ ^2E_S[g]\Big\} + \Big\{ \big\{^2E_S[g], \mathrm{Tr} \, h_{\gamma}  \big\}, ^2E_S[f] \Big\} \\
    &=  \Big\{  \big\{ \mathrm{Tr} \,h_{\gamma},\ ^2E_S[g]  \big\}, \ ^2E_S[f]\Big\} + \Big\{ \big\{^2E_S[g],\mathrm{Tr} \, h_{\gamma}  \big\}, ^2E_S[f] \Big\} \\
    &= \Big\{  \big\{\mathrm{Tr} \, h_{\gamma},\ ^2E_S[g]  \big\}, \ ^2E_S[f]\Big\} - \Big\{ \big\{ \mathrm{Tr} \,h_{\gamma}, ^2E_S[g]  \big\}, ^2E_S[f] \Big\}= 0 \ ,
\end{align*}
which was the expected result.

%
%

%
%
\section{Conclusions and comments}{\label{sec_conclusions}}

In this paper we have discussed the computation of Poisson brackets in a specific class of field theories in order to dispel some common misconceptions which originate in the idea (explicitly put forward by Dirac in \cite{Dirac_book}, but probably very old) that the transition from mechanics to field theory can be made by adapting in a straightforward way the usual concepts and formulas of the familiar Hamiltonian treatment of classical mechanics.

In order to convey our message in the clearest way we have employed very well-behaved functional spaces: Hilbert spaces of the Sobolev type. These spaces are helpful for several reasons:
\begin{itemize}
\item When they are used as the configuration manifolds for field theories, the canonical symplectic forms on the associated cotangent bundles are strongly non-degenerate. This means that the computation of the Hamiltonian vector fields defined by suitably regular functions is completely direct. The Poisson brackets are always well defined and well behaved.
\item As they are Banach spaces with the norm associated with the scalar product, the standard results about differentiability hold. This is important, especially for field theories defined on manifolds with boundary because, in that context, several definitions of differentiability can be found in the physics literature.
\item The embedding theorems guarantee that, with a judicious choice of Sobolev space, it is always possible to have sufficiently regular representatives of the basic fields. This makes it possible to define the evaluations of fields and momenta at spatial points and use these objects as the basic canonical variables in a rigorous way.
\item In many standard treatments of field theories and, in particular, when discussing the Hamiltonian formulation, the Poisson brackets of the fields and other objects are written in terms of distributions. Although distributions are, actually, very smooth objects (for instance, they can be differentiated any number of times), the topology of the spaces of distributions is not as simple as that of Banach and Hilbert spaces. Also, they have some unpleasant features, for instance, their multiplication is not always defined.
\end{itemize}

It is not clear to us that physically interesting field theories can be defined in the tangent and cotangent bundles of Sobolev spaces in a straightforward way. As discussed in \cite{Gotaythesis,ChernoffMarsden}, even for such simple examples as the free scalar field, the appropriate setting seems to be that provided by the so called \textit{manifold domains} in which the fibers are different from the configuration space. We nonetheless feel that some of our results on the computation of Poisson brackets in Sobolev spaces provide tantalizing hints about the possibility of avoiding the appearance of divergent quantities and ill defined objects such as the product of delta distributions.

In the context of LQG there have been numerous discussions over the years regarding the apparently counterintuitive non-commutativity of the flux variables. We have several comments on this issue. The first is that the standard flux-holonomy algebra is perfectly well defined and its properties understood in a satisfactory way (see \cite{Thiemann1,Lew1}). We are by no means suggesting here that it should be replaced by a ``naive'' algebra---which would, in fact, mimic the primitive proposals for quantum loop variables. Our point is, only, that it is not possible to justify the necessity of having non Poisson-commuting fluxes by appealing to some purported violation of the Jacobi identity. In fact, this violation originates in the non-(Fr\'echet) differentiability of the fluxes when they are modeled with the help of $L^2$ spaces. If the Poisson brackets can be computed, \textit{the Jacobi identity will always hold}, as we have illustrated with the help of the mock holonomy algebra discussed in Section \ref{sec_mock_HFalgebra}. By using suitable regularizations as in \cite{Thiemann1}, differentiability is restored and, hence, the computations leading to the Jacobi identity can be justified. The regularization process makes use of functions that depend both on the connection and the triad, so the non-commutativity of the fluxes is, actually, natural.

As a last comment, in our opinion, it may be worth studying the possibility of defining physically interesting field theories in the tangent and cotangent bundles of Sobolev Hilbert spaces, such as the ones used in this paper. This would make it possible to take advantage of the mathematical structures available in this setting and the useful properties of the objects that can be defined in these spaces. It is not clear to us that this can be done in the specific case of LQG, but it is an idea worth exploring.

\section*{Acknowledgments}
This work has been supported by the Spanish Ministerio de Ciencia Innovaci\'on y Uni\-ver\-si\-da\-des-Agencia Estatal de Investigaci\'on PID2020-116567GB-C22 grants. E.J.S. Villase\~nor is supported by the Madrid Government (Comunidad de Madrid-Spain) under the Multiannual Agreement with UC3M in the line of Excellence of University Professors (EPUC3M23), and in the context of the V PRICIT (Regional Programme of Research and Technological Innovation). Bogar D\'iaz acknowledges support from the CONEX-Plus programme funded by Universidad Carlos III de Madrid and the European Union's Horizon 2020 research and innovation programme under the Marie Sk{\l}odowska-Curie grant agreement No. 801538.

\appendix

%
%
\section{The Riesz-Fréchet representative of the evaluation}{\label{app_1}}

In this appendix we show how the Riesz-Fréchet representative of the evaluation $\Ev_{\!x}$ can be found. We will concentrate on the evaluation in $H^1(0,1)$, the other examples can be approached in an analogous way.

An orthonormal basis for $H^1(0,1)$ is
\begin{align*}
  & s_0 (t)=\frac{1}{\sqrt{\sinh1}}\sinh\left(t-\frac{1}{2}\right)\,, \\
  & c_0 (t) =1\,, \\
  & s_k (t)=\sqrt{\frac{2}{1+4\pi^2 k^2}}\sin(2\pi k t)\,,\quad k\in\mathbb{N}\,, \\
  & c_k (t)=\sqrt{\frac{2}{1+4\pi^2 k^2}}\cos(2\pi k t)\,,\quad k\in\mathbb{N}\,,
\end{align*}
with $t\in[0,1]$. In terms of this basis, an element $f\in H^1(0,1)$ has an expansion of the form
\begin{align*}
f(x)=&\frac{a_0}{\sqrt{\sinh 1}}\sinh\left(x-\frac{1}{2}\right)+b_0\\
&+\sum_{k=1}^\infty\sqrt{\frac{2}{1+4\pi^2 k^2}}\big(a_k\sin(2\pi k x)+b_k\cos(2\pi k x)\big)\,,
\end{align*}
with $H^1(0,1)$ norm given by
\begin{align*}
\|f\|^2_{H^1}=\sum_{k=0}^\infty(a_k^2+b_k^2)<+\infty\,.
\end{align*}
As a curious side remark, it is interesting to notice that the value of $a_0$ has the following simple expression, depending only on $\tilde{f}(0)$ and $\tilde{f}(1)$,
\begin{align*}
a_0=\sqrt{\frac{1}{2}\coth\frac{1}{2}}\big(\tilde{f}(1)-\tilde{f}(0)\big)\,.
\end{align*}
This means that the elements $f\in H^1(0,1)$ orthogonal to $s_0$ satisfy $\tilde{f}(1)=\tilde{f}(0)$.

According to \eqref{Riesz} the Riesz-Fréchet representative of the evaluation $\Ev_x$, $x\in[0,1]$ is given by
\begin{align*}
\mathcal{E}_x (t)=1+\frac{1}{\sinh1}\sinh\left(x-\frac{1}{2}\right)\sinh\left(t-\frac{1}{2}\right)+2\sum_{k=1}^\infty \frac{\cos\big(2\pi k(x-t)\big)}{1+4\pi^2k^2}\,,
\end{align*}
which can be shown to be equal to \eqref{Riesz_rep_eval}. Its operator norm is
\begin{align*}
\|\Ev_x\|=\|\mathcal{E}_x\|_{H^1}=\left(\frac{e^x+e^{2-x}}{e^2-1}\cosh x\right)^{1/2}\,.
\end{align*}

It is instructive to see now why the evaluation of the derivative is not a continuous linear functional in $H^1(0,1)$ by using \eqref{nec_cond_Riesz}. Indeed, if we denote the evaluation of the derivative as $\Ev_x'$ we have
\begin{align*}
\sum_{n\in\mathbb{N}}\big|\Ev_x'(e_n)\big|^2=\frac{1}{\sinh1}\cosh^2\left(x-\frac{1}{2}\right)+\sum_{k=1}^\infty\frac{8\pi^2k^2}{1+4\pi^2k^2}\,,
\end{align*}
which diverges (here we have denoted the $s_k$ and $c_k$, $k=0,1,\ldots$ collectively as $e_n$). Notice that the series
\begin{align*}
\sum_{n\in\mathbb{N}}\Ev_x'(e_n)e_n(t)=&\frac{1}{\sinh 1}\cosh\left(x-\frac{1}{2}\right)\sinh\left(t-\frac{1}{2}\right)\\
&+\sum_{k=1}^{\infty}\frac{4\pi k}{1+4\pi^2 k^2}\sin\big(2\pi k(t-x)\big)
\end{align*}
converges pointwise (as can be seen by using Dirichlet's convergence criterion) for every $t\in[0,1]$, but its sum has a jump discontinuity at $t=x$. For this reason it is not an element of $H^1(0,1)$.

%
%
\section{Details for the proof of Proposition \ref{11}}{\label{app_2}}

\subsection{The continuous inclusion \texorpdfstring{$H^2(\mathbb{R}^3)\hookrightarrow L^\infty(\mathbb{R}^3)$}{H2includedLinf}}\label{cont_inc}

$C_c^\infty(\mathbb{R}^3)$ is dense in $H^2(\mathbb{R}^3)=W^{2,2}(\mathbb{R}^3)$. This means that given $u\in H^2(\mathbb{R}^3)$ we can find a sequence of smooth functions of compact support $(\varphi_n)_{n\in \mathbb{N}}$, converging to $u$ in $H^2(\mathbb{R}^3)$. Obviously $C_c^\infty(\mathbb{R}^3)\subset  H^2(\mathbb{R}^3)$.

\medskip

As a consequence of Sobolev's embedding theorem we have $\mathbbm{H}^2(\mathbb{R}^3)\hookrightarrow L^\infty(\mathbb{R}^3)$ (i.e. $\mathbbm{H}^2(\mathbb{R}^3)$ is a vector subspace of $L^\infty(\mathbb{R}^3)$ and the inclusion map is continuous). This means that, given an open set $V$ in $L^\infty(\mathbb{R}^3)$, the set $V\cap \mathbbm{H}^2(\mathbb{R}^3)$ is open in $\mathbbm{H}^2(\mathbb{R}^3)$.

Let  $\varepsilon>0$ and $u\in\mathbbm{H}^2(\mathbb{R}^3)$ and consider the open ball $B^\infty(u;\varepsilon)$ in $L^\infty(\mathbb{R}^3)$. As the inclusion of $\mathbbm{H}^2(\mathbb{R}^3)$ in $L^\infty(\mathbb{R}^3)$ is continuous, the intersection $B^\infty(u;\varepsilon)\cap \mathbbm{H}^2(\mathbb{R}^3)$ is an open subset of $\mathbbm{H}^2(\mathbb{R}^3)$. This means that we can find an open ball $B^2(u;\varepsilon')\subset\mathbbm{H}^2(\mathbb{R}^3)$ (defined in terms of $\|\cdot\|_{\mathbbm{H}^2}$), with $\varepsilon'>0$, which is contained in $B^\infty(u;\varepsilon)$. Hence, if a sequence of elements of $\mathbbm{H}^2(\mathbb{R}^3)$ converges to $u\in\mathbbm{H}^2(\mathbb{R}^3)$ in the $\|\cdot\|_{\mathbbm{H}^2}$ norm, it also converges to $u$ in the $\|\cdot\|_\infty$ norm. Indeed, given $\varepsilon>0$ take the ball $B^\infty(u;\varepsilon)$ and find some $B^2(u;\varepsilon')$ as described above. The convergence of the sequence in $\mathbbm{H}^2(\mathbb{R}^3)$ implies the existence of $n_0\in \mathbb{N}$ such that if $n>n_0$ then $u_n\in B^2(u;\varepsilon')$ and, hence, $u_n\in B^\infty(u;\varepsilon)$.

\subsection{Supremum, essential supremum and continuous functions}\label{ess_sup}

Let $U$ be an open subset of $\mathbb{R}^n$ and $f:U\rightarrow \mathbb{R}$ a measurable function (with respect to the Lebesgue measure). For $a\in \mathbb{R}$ we have
\[
f^{-1}(a,\infty)=\{x\in U:f(x)>a\}\,.
\]
Notice that $f^{-1}(a,\infty)=\varnothing$ is equivalent to $a\geq f(x)$ for all $x\in U$, i.e. $a$ is an upper bound of $f(U)$. We define now
\begin{align*}
\sup_{x\in U}f&:=\inf\{a\in \mathbb{R}:f^{-1}(a,\infty)=\varnothing\}\,,\\
\mathrm{ess}\,\sup_{x\in U}f&:=\inf\{a\in \mathbb{R}:\mu(f^{-1}(a,\infty))=0\}\,.
\end{align*}
Notice that $\mu(\varnothing)=0$ implies
\[
\{a\in \mathbb{R}:f^{-1}(a,\infty)=\varnothing\}\subset\{a\in \mathbb{R}:\mu(f^{-1}(a,\infty))=0\}\,,
\]
and, hence,
\[
\mathrm{ess}\,\sup_{x\in U}f\leq \sup_{x\in U}f\,.
\]

\noindent Let us prove now the following

\begin{lemma}
Let $U\subset \mathbb{R}^n$ be open and $f:U\rightarrow \mathbb{R}$ \emph{continuous}, then $\sup_{x\in U}f=\mathrm{ess}\sup_{x\in U}f$.
\end{lemma}

\begin{proof} We just have to show $\sup_{x\in U}f\leq \mathrm{ess}\,\sup_{x\in U}f$. We will do it by contradiction. Let us then suppose
\[
\sup_{x\in U}f> \mathrm{ess}\,\sup_{x\in U}f\,.
\]
This implies the existence of $a\in \mathbb{R}$ such that $\mathrm{ess}\,\sup_{x\in U}f<\!a<\!\sup_{x\in U}f$. Now, $a\!<\!\sup_{x\in U}f$ implies $f^{-1}(a,\infty)\neq\varnothing$ so there is some $x_0\in U$ satisfying $a<f(x_0)\leq \sup_{x\in U}f$. Let $\varepsilon=(f(x_0)-a)/2$, then, as $f$ is continuous, we can find $\delta>0$ such that, for all $x\in B(x_0;\delta)\cap U$ we have $|f(x)-f(x_0)|<\varepsilon$ and, thus, $f(x)>a+\varepsilon$.

As a consequence, $\mu(f^{-1}(a+\varepsilon,\infty))>0$ (notice that for this to be true it is crucial that $U$ is open). This implies that $a+\varepsilon\!\leq\!\mathrm{ess}\,\sup_{x\in U}f$ [because $\alpha\!>\!\mathrm{ess}\,\sup_{x\in U}f\Rightarrow \mu(f^{-1}(\alpha,\infty))=0$ is equivalent to $\mu(f^{-1}(\alpha,\infty))\!>\!0\Rightarrow \alpha\leq \mathrm{ess}\,\sup_{x\in U}f$] and thus $a<\mathrm{ess}\,\sup_{x\in U}f$, which is  impossible because $a>\mathrm{ess}\,\sup_{x\in U}f$. 
\end{proof}

\noindent \textbf{Remark:} The hypothesis $U$ open is fundamental as shown by the following example: The function
\[
\chi:U\rightarrow \mathbb{R}:x\mapsto x
\]
with $U=\mathbb{Q}\cap(0,1)$ is continuous with the natural topologies and
\[
\sup_{x\in U} \chi =\mathrm{inf}\{a\in\mathbb{R}:f^{-1}(a,\infty)=\varnothing\}=1\,,
\]
but
\[
\mathrm{ess}\sup_{x\in U} \chi =\mathrm{inf}\{a\in\mathbb{R}:\mu\big(f^{-1}(a,\infty)\big)=0\}
\]
does not exist because $\mu(f^{-1}(a,\infty))=0$ for all $a\in\mathbb{R}$.

%
%
\section{Canonical variables in loop quantum gravity}{\label{app_3}}

The symplectic form of LQG can be written in terms of a coframe $e_i$ and an $SO(3)$ connection $A_i$ as
\[
\Omega=\int_\Sigma \dd\left(\epsilon^{ijk}e_j\wedge e_k\right)\ww \dd A_i=:\int_\Sigma \dd E^i\ww \dd A_i
\]
(see \cite{concise}) with $E^i\in\Omega^2(\Sigma)$. In order to get a vector field from $E^i$ we introduce a volume form. Two possible ways to do this are:
\begin{itemize}
\item Introduce a fiducial volume form $v$.
\item Build a volume form $w$ from $E^i$.
\end{itemize}
In both cases we define $\hat{E}^i$ and $\tilde{E}^i$ as the dual-dual objects
\[
\hat{E}^i=\left(\frac{\cdot\wedge E^i}{v}\right)\,,\quad \tilde{E}^i=\left(\frac{\cdot\wedge E^i}{w}\right)\,,
\]
and work with the vector fields canonically defined by them. Here, for a top form $\beta$ and a volume form $v$ we write
\[
\beta=\left(\frac{\beta}{v}\right)v. 
\]
If $\alpha\in\Omega^1(\Sigma)$ these satisfy
\[
\alpha(\hat{E}^i)=\left(\frac{\alpha\wedge E^i}{v}\right)=\imath_{\hat{E}^i}\alpha\,,\quad \alpha(\tilde{E}^i)=\left(\frac{\alpha\wedge E^i}{w}\right)=\imath_{\tilde{E}^i}\alpha\,.
\]
The relationship between these fields is simple. As
\[
\alpha(\tilde{E}_i)=\left(\frac{\alpha\wedge E^i}{w}\right)=\left(\frac{\alpha\wedge E^i}{v}\right)\left(\frac{v}{w}\right)=\alpha(\hat{E}^i)\left(\frac{v}{w}\right)\,,
\]
for any 1-form $\alpha\in\Omega^1(\Sigma)$, we must have
\[
\tilde{E}^i=\hat{E}^i\left(\frac{v}{w}\right)\,.
\]
This implies
\[
\imath_{\tilde{E}^i}w=\imath_{\hat{E}^i}v\,.
\]
As a consequence, all the dependence in the phase space variables can be concentrated on the $\hat{E}^i$ (the volume form is field-independent).

It is interesting to find out the relationship between the divergencies of the vector fields $\hat{E}^i$ and $\tilde{E}^i$:
\[
\mathrm{div}_w\tilde{E}^i=\left(\frac{\mathrm{d}\imath_{\tilde{E}^i}w}{w} \right)=\left(\frac{\mathrm{d}\imath_{\hat{E}^i}v}{w}\right)=\left(\frac{\mathrm{d}\imath_{\hat{E}^i}v}{v}\right)\left(\frac{v}{w}\right)=\mathrm{div}_v\hat{E}^i\left(\frac{v}{w}\right)\,. 
\]
Notice that, as a consequence of this, it is possible to write the constraints in the Ashtekar formulationn in terms of either $\tilde{E}^i$ or $\hat{E}^i$ while keeping their form, for instance, the Gauss law is
\[
0=\mathrm{div}_w\tilde{E}^i+\epsilon^i_{\phantom{i}jk}\imath_{\tilde{E}^k}A^j=\big(\mathrm{div}_v\hat{E}^i+\epsilon^i_{\phantom{i}jk}\imath_{\hat{E}^k}A^j\big)\left(\frac{v}{w}\right)\,.
\]

\section*{References}
\bibliographystyle{iopart-num}
\bibliography{draft}

\providecommand{\newblock}{}
\begin{thebibliography}{10}
\expandafter\ifx\csname url\endcsname\relax
  \def\url#1{{\tt #1}}\fi
\expandafter\ifx\csname urlprefix\endcsname\relax\def\urlprefix{URL }\fi
\providecommand{\eprint}[2][]{\url{#2}}

\bibitem{Ashtekar1}
Ashtekar A 1986 {\em Phys. Rev. Lett.\/} {\bf 57}(18) 2244--2247

\bibitem{Ashtekar2}
Ashtekar A 1987 {\em Phys. Rev. D\/} {\bf 36}(6) 1587--1602

\bibitem{ACZ}
Ashtekar A, Corichi A and Zapata J~A 1998 {\em Classical and Quantum Gravity\/}
  {\bf 15} 2955--2972

\bibitem{HFA}
Rovelli C and Smolin L 1990 {\em Nucl. Phys. B\/} {\bf 331} 80--152

\bibitem{Thiemann1}
Thiemann T 2001 {\em Classical and Quantum Gravity\/} {\bf 18} 3293--3338

\bibitem{CattaPer}
Cattaneo A~S and Perez A 2017 {\em Classical and Quantum Gravity\/} {\bf 34}
  107001

\bibitem{Freidel_2013}
Freidel L, Geiller M and Ziprick J 2013 {\em Classical and Quantum Gravity\/}
  {\bf 30} 085013

\bibitem{Brezis}
Brezis H 2010 {\em Functional Analysis, Sobolev Spaces and Partial Differential
  Equations\/} Universitext (Springer New York) ISBN 9780387709130

\bibitem{Evans}
Evans L 2010 {\em Partial Differential Equations\/} second edition ed Graduate
  studies in mathematics (American Mathematical Society) ISBN 9780821849743

\bibitem{RT}
Regge T and Teitelboim C 1974 {\em Annals of Physics\/} {\bf 88} 286--318 ISSN
  0003-4916

\bibitem{soloviev}
Soloviev V~O 1997 {\em Phys. Rev. D\/} {\bf 55}(12) 7973--7976

\bibitem{Faddeev}
Faddeev L~D and Takhtajan L~A 2007 {\em {H}amiltonian {M}ethods in the {T}heory
  of {S}olitons\/} Classics in Mathematics (Springer Berlin, Heidelberg)

\bibitem{ChernoffMarsden}
Chernoff P and Marsden J 1974 {\em {P}roperties of {I}nfinite {D}imensional
  {H}amiltonian {S}ystems\/} Lecture Notes in Mathematics (Springer Berlin,
  Heidelberg)

\bibitem{Marsden}
Marsden J 1974 {\em Applications of Global Analysis in Mathematical Physics\/}
  Berkeley mathematics lecture notes (Publish or Perish, Incorporated) ISBN
  9780914098119

\bibitem{Adams}
Adams R~A and Fournier J~J 2003 {\em Sobolev spaces\/} Pure and Applied
  Mathematics (Elsevier) ISBN 9780120441433

\bibitem{Baez}
Baez J and Muniain J~P 1994 {\em Gauge {F}ields, {K}nots and {G}ravity\/}
  (World Scientific)
  \urlprefix\url{https://www.worldscientific.com/doi/abs/10.1142/2324}

\bibitem{Lew1}
Lewandowski J, Newman E~T and Rovelli C 1993 {\em Journal of Mathematical
  Physics\/} {\bf 34} 4646--4654

\bibitem{GROSS19851}
Gross L 1985 {\em Journal of Functional Analysis\/} {\bf 63} 1--46 ISSN
  0022-1236

\bibitem{Dirac_book}
Dirac P 2013 {\em Lectures on Quantum Mechanics\/} Dover Books on Physics
  (Dover Publications) ISBN 9780486320281

\bibitem{Gotaythesis}
Gotay M~J 1979 {\em Presymplectic manifolds, geometric constraint theory and
  the Dirac-Bergmann theory of constraints\/} Ph.D. thesis Center for
  Theoretical Physics of the University of Maryland

\bibitem{concise}
Barbero~G J~F, D\'{\i}az B, Margalef-Bentabol J and Villase\~nor E~J~S 2021
  {\em Phys. Rev. D\/} {\bf 103}(2) 024051

\end{thebibliography}
\end{document}